\documentclass[12pt]{amsart}%
\usepackage[T1]{fontenc}
\usepackage{amsmath, amsthm, amsfonts, amssymb}
\usepackage{mathrsfs}
\usepackage{graphicx}
\usepackage{url}
\usepackage{enumerate}
\usepackage{fancyhdr}
\usepackage{xspace}
\usepackage{verbatim}

\usepackage[margin=1.00in]{geometry}%
\usepackage{hyperref}

\theoremstyle{definition}

\newtheorem{defn}{Definition}[section]

\theoremstyle{remark}
\newtheorem{rmk}[defn]{Remark}
\newtheorem*{remark*}{Remark}

\theoremstyle{plain}

\newtheorem{cor}[defn]{Corollary}

\newtheorem{lem}[defn]{Lemma}

\newtheorem{prop}[defn]{Proposition}

\newtheorem{thm}[defn]{Theorem}

\numberwithin{equation}{section}
\numberwithin{figure}{section}
\newcommand{\paren}[1]{\left(#1\right)}
\newcommand{\set}[1]{\left\{#1\right\}}
\newcommand{\Z}{\mathbb{Z}}
\newcommand{\R}{\mathbb R}
\newcommand{\C}{\mathbb C}
\newcommand{\N}{\mathbb N}
\newcommand{\spec}{\sigma}
\newcommand{\Lap}{\mathcal{L}}
\newcommand{\spn}{\operatorname{span}}
\newcommand{\eng}{\mathscr{E}}

\newcommand{\alb}{\set{1,2,\ldots,3N}}

\newcommand{\eps}{\varepsilon}

\newcommand{\resscaling}{\rho}

 \def\mult#1#2{\textup{mult}_{#1}{(#2)}}

\def\s-s{self-similar}

\DeclareMathSizes{10}{9}{7}{5}

\allowdisplaybreaks[4]
\title{Gaps in the  spectrum of the Laplacian on $3N$-Gaskets}

\author[D.~Kelleher]{Daniel ~Kelleher}
\address{Department of Mathematics, Purdue University, West Lafayette, IN,  47907, USA }
\email{dkellehe@purdue.edu}
\urladdr{\url{http://www.math.purdue.edu/~dkellehe/}}

\author[N.~Gupta]{Nikhar Gupta}\email{Nikhaar.Gupta@uconn.edu}
\author[M.~Margenot]{Maxwell Margenot}\email{mmargeno42@gmail.com}
\author[J.~Marsh]{Jason Marsh}\email{jmarsh2@nd.edu}
\author[W.~Oakley]{William Oakley}\email{wgoakley@math.ucla.edu}
\author[A.~Teplyaev]{Alexander Teplyaev}\email{teplyaev@math.uconn.edu}
\address{Department of Mathematics, University of Connecticut, Storrs, CT 06269, USA}
\urladdr{\url{http://www.math.uconn.edu/~teplyaev/}}

\thanks{Authors supported in part by the National Science Foundation through grant DMS-1106982.}
\begin{document}
%
\begin{abstract}
This article develops analysis on fractal $3N$-gaskets, a class of post-critically finite fractals which include the Sierpinski triangle for $N=1$, specifically properties of the Laplacian $\Delta$ on these gaskets. We first prove the existence of a self-similar geodesic metric on these gaskets, and prove heat kernel estimates for this Laplacian with respect to the geodesic metric.  We also compute the elements of the method of spectral decimation, a technique used to determine the spectrum of post-critically finite fractals. Spectral decimation on these gaskets arises from more complicated dynamics than in previous examples, i.e. the functions involved are rational rather than polynomial. Due to the nature of these dynamics, we are able to show that there are gaps in the spectrum.
\end{abstract}
%
\date{\today}
\maketitle
%
\tableofcontents
%
\section{Introduction} \label{introduction}

Recently there has been considerable interest in Laplace operators with exotic spectral properties. 
For instance, large spectral gaps  imply uniformly convergent Fourier expansions on fractals, see \cite{Str05}, and  oscillations  in the spectrum, see \cite[and references therein]{ADT09,ADTV,HSTZ,HZ09,hexa,Kajino10,Kajino13,Kajino13b,Kajino14,KL93}. 
However there had not been any infinite families of fractals with complicated spectra which can be computed. It is especially interesting to have families of  ``fractal strings'' (the term introduced by M.~Lapidus, see \cite{Lap1,Lap2}) that converge to classical loops as $N\to\infty$ and have meromorphic spectral zeta functions (see \cite{DGV08,DGV12,LL12,StTzeta,Tep07}). 
Such families would be of particularly significance  in view of the appearance of fractal models in the 
theory of Quantum Gravity, as in the works of F.~Englert, R.~Loll, M.~Reuter et al., see \cite[and references therein]{Englert,Loll,Reuter} (see \cite{ADT09,ADT10,Dun12,KKP+12} for most immediate applications in Physics, such as 
 oscillations of the heat kernel or the Casimir effect). 
In our work we study the first  family of fractals known in the literature that combines all these aspects of the general theory.   

Our main geometric objects of study, the $3N$-gaskets, are post critically finite self-similar fractals with the symmetries of a regular $3N$-gon. They were introduced in \cite{TW05,TW06} and their resistance was studied in \cite{BCF+07} (in more general framework the existense and uniqueness of energy and Laplacians are discussed in \cite{BBKT,HMT06,KSW,Pei08,Pei14a,Pei14b,T08cjm} and references therein). 
In this paper we show how to compute the spectrum of the Laplacian operator for this infinite family of fractals using the technique of spectral decimation. Roughly speaking,  one begins by calculating the dynamics map $R(z)$, which maps the spectrum of the Laplacian on the $n$-th approximating graph to the spectrum of the graph Laplacian on a coarser $(n-1)$-th approximation, i.e. $R:\sigma(\Delta_n)\to\sigma(\Delta_{n-1})$. This allows one to identify the spectrum of the Laplacian on the limit of these graphs as   the Julia set of $R(z)$ and an infinite sequence of isolated eigenvalues accumulating to this Julia set. 
The basics techniques and theory of analysis on post critically finite fractals, such as $3N$-gaskets, can be found \cite{Str06,Kig93,Kig03}, including the construction of the Laplacian and energy form. The textbook \cite{Str06} contains comprehensive explanation of the spectral decimation method in the case of the Sierpinski gasket. 
 To compute the multiplicities of eigenvalues, we use \cite{Tep08b,Tep08a}, and also prove a new general result on spectral decimation, Theorem~\ref{thm-new}.

Historically, the spectral decimation was first used to calculate the spectrum of the Sierpinski gasket fractal in \cite{FS92, Shi91}, though it was framed in a different manner than it is here. Older works from physics literature, \cite{Ale84,BKP85,Bel92,R84,RT82} used similar techniques to calculate the spectrum of Schr\"odinger type operators on  Sierpinski  lattices. 
It was shown that the general technique of spectral decimation works for a certain types of fractal with the right symmetry conditions \cite{MT95,Tep98,MT03,DB10,DS09}. In \cite{Tep08a,Tep08b}, several examples are worked out in detail, including of the 6-gasket which is also covered in the current work. It was also used in to find the spectrum of infinite Sierpinski fractafolds in \cite{ST12}.
However, the class of examples for which the exact algorithm is known is limited, and in known examples the dynamics which give rise to these algorithms are relatively simple in that they involve only iterating polynomials or low degree rational functions. This paper presents the first infinite class of fractals for which the dynamics arises from iterations of  unboundedly high degree rational functions.

 Concerning further research,     \cite{ACSY14}  discusses how to develop a theory of 1-forms from discrete 1-forms in such a way that is applicable to other p.c.f.-fractals, in particular $3N$-gaskets because they share symmetries with the Sierpinski Gasket. 
 The general theory of 1-forms on fractals and related questions are discussed in   \cite{IRT12,h5,h4,h3,h2,h1,H16,HKT13,HR16,HT15}, and related resolvent estimates and distribution theory in \cite{r5,r4,r3,r2,r1}. 

\subsection*{Acknowledgments}
 The first and the last authors are very grateful to Robert Strichartz for interesting discussions and many helpful suggestions. 
 
%
\section{Definitions and main results}\label{results}\label{preliminaries}
For a fixed natural number $N>0$, a $3N$-Gasket is a post-critically finite self-similar fractal $K$ with a set of $3N$ maps $F_i: K \to K$, $i=1,2, \ldots, 3N$. $K$ has fractal boundary $V_0=\set{v_1,v_2,v_3}$, and the maps $F_i$ are determined by the rule $F_i(v_2) = F_j(v_3)$ where $i= j+1 \mod 3N$ and $F_{kN}(v_1)= v_k$ for $k=1,2,3$. 

This defines a self-similar structure in the sense of \cite{Kig93,Kig01} by the theory of ancestors in \cite{Kig93}. The limiting self-similar structure is $K$ is defined as the quotient of the shift space
\[
\Sigma = \set{1,2,\ldots,3N}^\N
\]
of infinite words on the alphabet of length $3N$ with quotient map  $\pi: \Sigma \to K$. Here $v_1  = \pi(\overline{N})$ where $\overline N := NNN\cdots$ is the infinite string of $N$'s, and $F_i(v_1) = \pi\paren{i\overline{N}}$. Thus $\pi$ is defined by the equivalence relation
\[
\pi\paren{ w(i)(2N)\overline{N}} = \pi\paren{ w (j) (3N)\overline{N}}
\]
where $w\in \set{1,2,\ldots, 3N}^k$ is a finite word and $i\equiv j+1 \mod 3N$.

\begin{remark*} 
In this paper we consider $3N$-gaskets as abstract \mbox{p.c.f.} fractals in the sense of Kigami, 
defined as quotient spaces of the infinite product space $\Sigma$. 
In the case where $N=1,2,3$ 
these fractals can be naturally embedded in 
$\R^2$ as so called regular polygaskets. 
These embeddings  are the  attractors of the iterated function system consisting of the $3N$ maps $G_i:\R^2\to\R^2$, $G_i(x) = r(x-p_i)+p_i$, where $p_i$ are the corners of a regular $3N$-gon,
\[
r = \frac{\sin(\pi/3N)}{\sin(\pi/3N) + \sin(\pi/3N + 2\pi m/3N)},
\]
and $m= \lfloor 3N/4\rfloor$. 
These fractals are nested fractals and are shown in Figure~\ref{fig1}. 
 For  $N\geqslant4$ the attractors of this 
iterated function system will not be the $3N$-gaskets, but different fractals 
(for istance, 
in the cases where $N$ is divisible by $4$, the attractor of the above iterated function system 
 is an infinitely ramified fractal, 
 and does not fit into the framework of the current paper). 
However
 all the $3N$ gaskets can   be embedded in 
$\R^2$ as 
self-similar sets in the sense of Euclidean self-similarity, {i.e.} using Euclidean contractive similitudes.  For instance, 
Figure~
\ref{SawtoothSetup}
shows a rough scheme how to construct an $\R^2$-embedding the $3N$-gasket for $N=5$. 
 In any case, the formal definition given before this remark is   valid for all natural $N$ and, 
 for the purpose of this paper, the embedding in $\R^2$ is not important. 
\end{remark*}

\begin{figure}
\includegraphics[width=.3\textwidth]{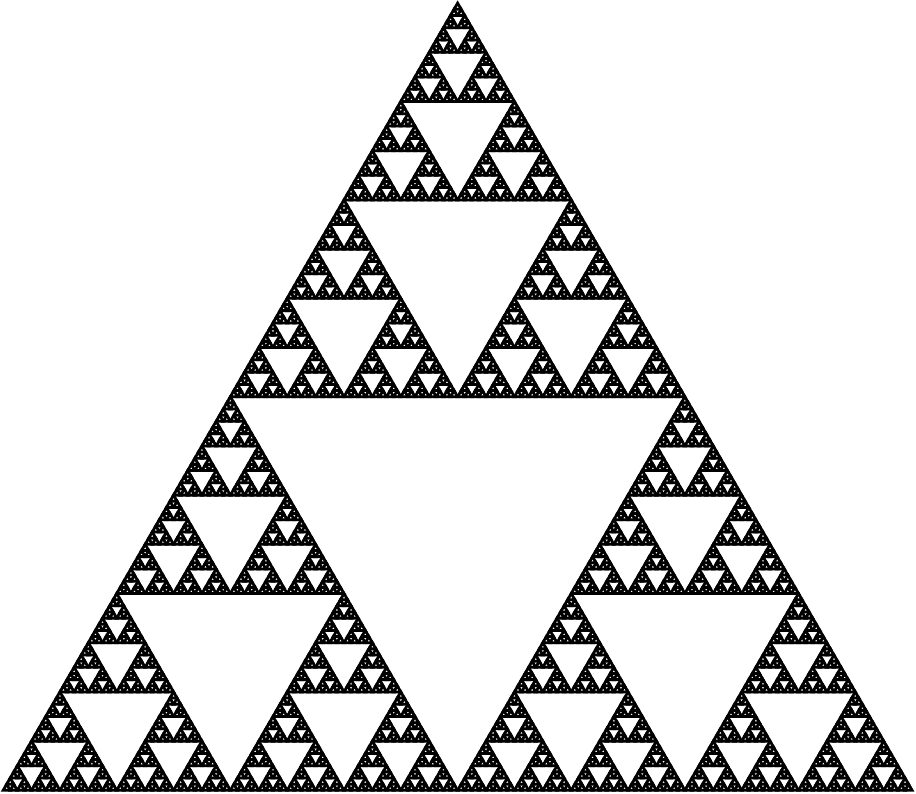}\hspace{.03\textwidth}
\includegraphics[width=.27\textwidth]{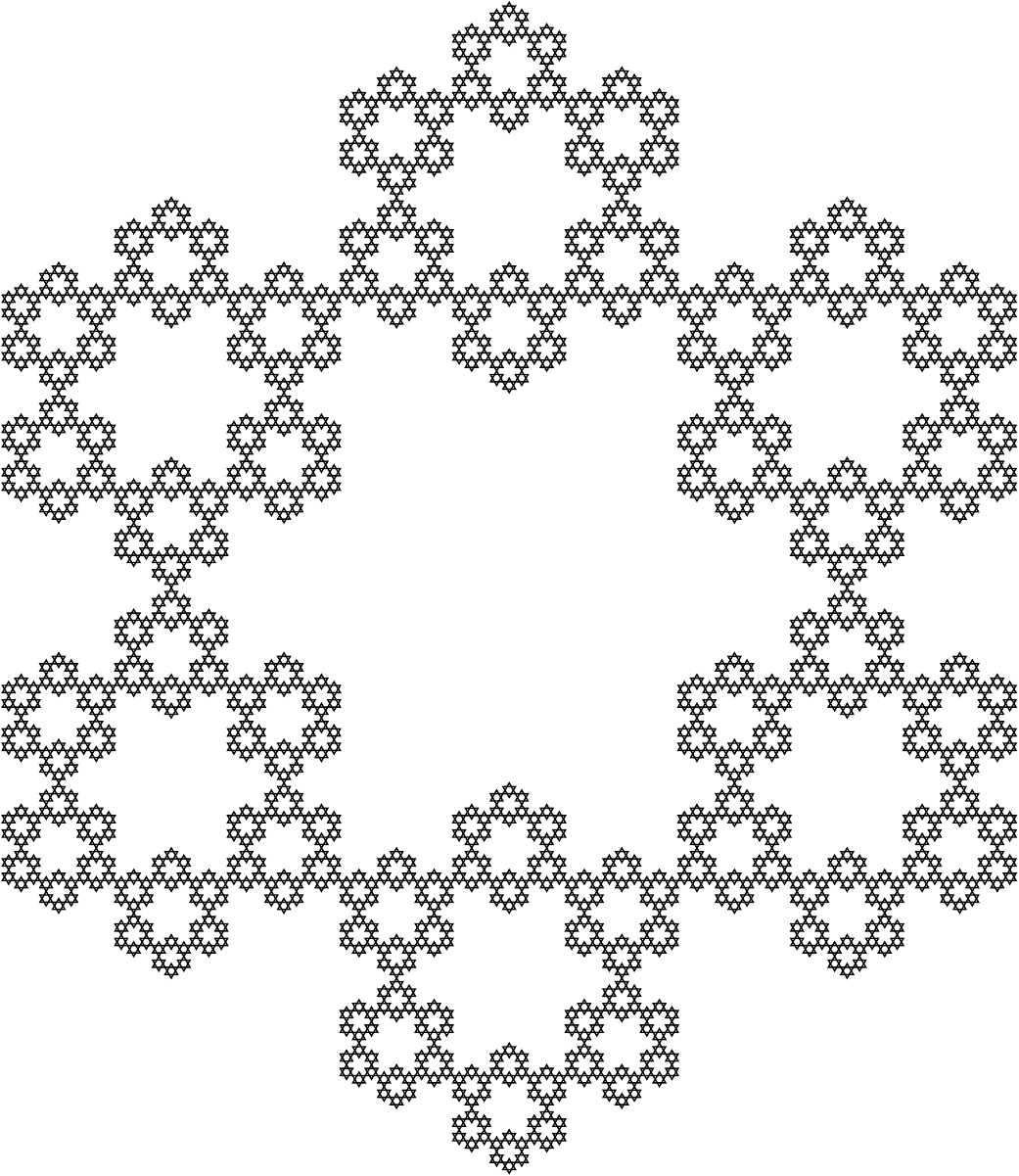}\hspace{.03\textwidth}
\includegraphics[width=.3\textwidth]{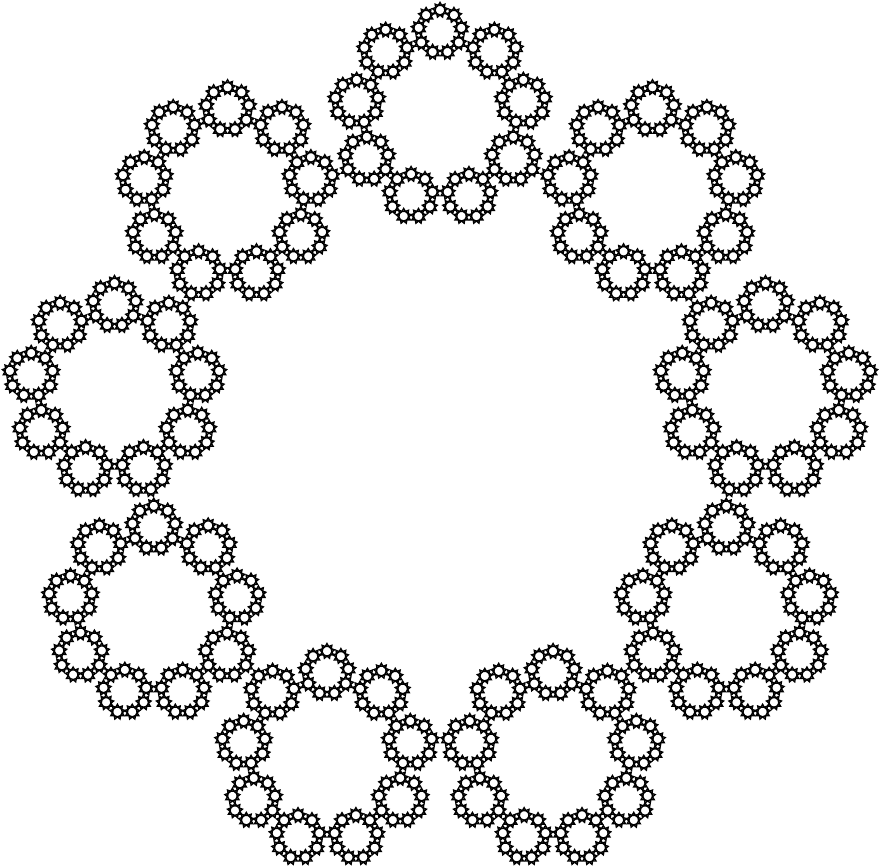}
\caption{The 3-gasket (Sierpinski gasket), 6-gasket (Hexagasket), and 9-gasket.}\label{fig1}
\end{figure}

We intrinsically approximate $K$ by the finite sets $V_n\subset K$, $n=0,1,2,\ldots$, where
\[
V_{n+1} = \bigcup_{i=1}^{3N} F_i(V_n),
\]
and $V_0$ is the fractal boundary.  
We think of $V_n$ as the vertex set of a graph, where $x,y\in V_n$ are adjacent if there is  $w = i_1i_2\cdots i_n\in\set{1,2,\ldots,3N}^n$, such that $
x$ and $y\in F_w(V_0)$. In this case we write $x\sim y$. Here $ F_w := F_{i_1}\circ F_{i_2}\circ\cdots\circ F_{i_n}$.

For a set $V$, let $\ell(V)= \set{f:V\to\R}$. If $V$ is finite $\ell(V) = \R^{|V|}$. We define the Laplacian on a given $3N$-gasket as the scaled limit of graph Laplacian on the sets $V_n$. Specifically, we define the operator $\Delta_n: \ell(V_n)\to\ell(V_n)$ by
\[
\Delta_n f(x) = f(x)-\frac{1}{d_x}\sum_{y\sim x} f(y),
\]
where $d_x $ is the degree of $x$ in the graph above. For $f\in \ell(K)$, and $x\in V_k$ for some $k$, define
\begin{align}\label{LaplacianDefn}
\Delta f(x) = \lim_{n\to\infty} c^n\Delta_n f|_{V_n}(x),
\end{align}
where
\begin{equation}\label{e-def-c}
	c = 3N+2N^2.
\end{equation}
We define the (continuous) domain of $\Delta$ to be the functions $f\in C(K)$ for which the above limit exists, and is uniform in  $x\in V^*$ and,  moreover, represents a continuous function on $V^*$. 
Then, in this case, $\Delta f(x)$ can be continued from $V^*$ to   a continuous function on the fractal. 
It follows from the theory of 
Kigami (\cite{Kig01})  that the limit is non trivial and has a unique extension to a self-adjoint operator, which we also denote $\Delta$, with discrete spectrum $\sigma(\Delta)$.

 Since $\ell(V_n)$ is a finite dimensional vector space for all $n$, $\Delta_n$ can be thought of as a matrix, and the spectrum $\spec(\Delta_n)$ is the eigenvalues of this matrix. For example, for any $N$, 
\[
\Delta_0 = \paren{\begin{array}{ccc}1& -1/2 & -1/2 \\ -1/2&1&-1/2\\-1/2&-1/2&1\end{array}},
\]
and thus $\spec(\Delta_0) = \set{0,3/2}$.

Take $d_n$ to be the graph distances on the $n$th level approximating graphs, that is
\begin{equation}\label{e-dn}
d_n(x,y) = \inf\set{k\in \N~|~\exists\set{x_i}_{i=0}^k\subset V_n,\text{ with } x=x_0\sim x_1 \sim\cdots\sim x_k=y}
\end{equation}
We prove in Section~\ref{metric} that when scaled properly, the sequence $\set{d_n}_{n=1}^\infty$ converges to $d_*$ which can be extended to a self-similar geodesic distance on the whole of the $3N$-gasket $K$ and induces the standard topology. 

Further, let $\mu$ be the symmetric self-similar measure on $K$, that is $\mu$ is the unique measure such that $\mu(F_w(K)) = (3N)^{-k}$ for all $w\in \set{1,\ldots,3N}^k$. Alternatively, if $\nu$ is the standard Bernoulli probability measure on $\set{1,\ldots, 3N}^{\N}$, then $\mu(U) = \nu(\pi^{-1}(U))$ for all Borel subsets $U$ of $K$.

The heat kernel $p(t,x,y): \R^+\times K\times K \to \R$ of the operator $\Delta$ with respect to $\mu$ is the integral kernel such that, if $f\in L^2(K)$ and $u: \R^+ \times K \to \R$ is defined
\[
u(t,x)  = \int p(t,x,y)f(y) \ d\mu(y)
\]
then $u$ satisfies the differential equation
\[
\begin{cases}
\Delta u(t,x) = \displaystyle\frac{du}{dt}(t,x) \\[8pt]
u(0,x) = f(x)
\end{cases}
\]
With respect to the distance $d_*$ and the measure $\mu$, we get the following heat kernel estimate
\begin{thm}\label{HKEthm}
The  Laplacian $\Delta$ on $K$ has a symmetric heat kernel $p(t,x,y): \R^+\times K\times K \to \R$ which satisfies the following estimates
\begin{align}
p(t,x,y) \leq c_1 t^{-d_H/d_w}\exp\paren{c_2 (d_*(x,y)^{d_w}/t)^{1/(d_w - 1)}}
\end{align}
and
\begin{align*}
p(t,x,y) \geq c_3 t^{-d_H/d_w}\exp\paren{c_4 (d_*(x,y)^{d_w}/t)^{1/(d_w - 1)}}. 
\end{align*}
Here $d_H = \log (3N)/\log(N+1)$ is the Hausdorff dimension with respect to $d_*$, 
 $d_w = \alpha + d_H=(\log(2N+3)+log(N))/\log(N+1)$
where $\alpha = \log \rho /\log(N+1)$,  $\rho = c/(3N)=1+2N/3$, and $c_1$,$c_2$,$c_3$,$c_4$ are positive real constants. 
\end{thm} 

The proof of this theorem is given in Section~\ref{metric}. We note here that the symmetry of the heat kernel is a result of the fact that $\Delta$ is self-adjoint with respect to the self-similar measure $\mu$. This follows from \cite{Kig01}. 

\begin{rmk}Note that the spectral dimension $$d_s=2d_H/d_w=2\log(3N)/(\log(2N+3)+log(N))$$ 
was first computed in \cite{BCF+07} in different notation. Our current $\rho$ was denoted by $c$ in \cite{BCF+07}, and our current $N$ was denoted by $3N$ in \cite{BCF+07}. According to the terminology of \cite{BCF+07}, our $3N$ gasket is a $(N,N,N)$-gasket, which means $N_1=N_2=N_3=N$ (but in \cite{BCF+07} the notation was $N=N_1+N_2+N_3$, which would contradict our use of $N$).
\end{rmk}

We use the technique called spectral decimation to calculate the spectrum of $\Delta_{n+1}$ from the spectrum of $\Delta_n$. Thinking of $\Delta_n$ as a matrix, the first step in spectral decimation is computing the rational functions $R(z)$ and $\phi(z)$ such that
\[
S(z) = \phi(z)(\Delta_n - R(z)).
\]
where the matrix-valued function $S(z) = (A-z) + B(D-z)^{-1}C$ is the Schur complement of 
\[
\Delta_{n+1} - z = \begin{pmatrix}
A-z & B \\ C & D-z
\end{pmatrix}
\] 
with block representation such that $A$ corresponds to the vertices in the $n$th level approximation. 
\begin{thm}\label{Randphi}
For the $3N$-gasket, $N\geq 1$, $R(z)$ and $\phi(z)$ are the rational functions
\begin{align}
R(z) = 
\begin{cases}
\frac{(z-1) \sqrt{z} U_{N-1}\left(\sqrt{z}\right) \left(2 T_N(1-2 z)+2 U_{N-1}(1-2 z)+1\right)}{T_N\left(\sqrt{z}\right)} & \text{if }N\text{ is even} \\[8pt]
\frac{\sqrt{z} T_N\left(\sqrt{z}\right) \left(2 T_N(1-2 z)+2 U_{N-1}(1-2 z)+1\right)}{U_{N-1}\left(\sqrt{z}\right)} & \text{if }N\text{ is odd},
\end{cases} \label{form1}
\end{align}
and
\begin{align}
\phi(z) = 
\begin{cases}
\frac{(3-2 z) T_N\left(\sqrt{z}\right)}{\left(T_N\left(\sqrt{z}\right)-2 (z-1) \sqrt{z} U_{N-1}\left(\sqrt{z}\right)\right) \left(2 T_N(1-2 z)+2 U_{N-1}(1-2 z)+1\right)}  & \text{if }N\text{ is even} \\[8pt]
\frac{(3-2z) U_{N-1}\left(\sqrt{z}\right)}{\left(U_{N-1}\left(\sqrt{z}\right)- 2 \sqrt{z} T_N\left(\sqrt{z}\right)\right) \left(2 T_N(1-2 z)+2 U_{N-1}(1-2 z)+1\right)} & \text{if }N\text{ is odd}.
\end{cases} \label{phi}
\end{align}
Further, for $N>1$, $R(z)$ has simple poles at the set of points
\begin{align}\label{sing}
&\set{\cos^{2}\paren{\frac{m\pi}{2N}}:m=1,3,\ldots, N-1} & \text{if }N\text{ is even, and}  \\
&\set{\cos^{2}\paren{\frac{m\pi}{2N}}:m=2,4,\ldots, N-1} & \text{if }N\text{ is odd}.
\end{align}
For $N>1$, the set of points above along with the point $3/2$ are the complete set of zeros of $\phi(z)$.
\end{thm}

This is proved in the appropriate subsections of Section~\ref{computation}. 

Above $T_k$ and $U_k$ are Chebyshev polynomials of the first and second type respectively, that is to say that
\[
T_k(z) = \cos( k \arccos(z)) \quad\text{and}\quad U_k (z) = \frac{\sin((k+1)\arccos(z))}{\sin(\arccos(z))}
\]
for the appropriate domains, and are extended by the polynomial representation everywhere else.

An interesting feature of this result is that the function $R(z)$ is given by a rational function which has disconnected real Julia set (see \cite{HSTZ,HZ09}) because $R(z)$ has poles in the convex hull of its Julia set. One implication of this is that the spectrum of the renormalized limit of $\Delta_n$, $\Delta$ has gaps in the sense of \cite{Str05}, which implies that the Fourier series of continuous functions converge uniformly on $K$.

\begin{thm}\label{gapsthm}
If $\set{\lambda_i}_{i=0}^\infty$, is the nondecreasing enumeration of the eigenvalues of $\Delta$, then $\limsup \lambda_{k+1}/\lambda_k > 1$.
\end{thm}

Knowing the functions $R(z)$ and $\phi(z)$ allow us to calculate the spectrum of $\Delta_n$  as follows.

\begin{thm}\label{finitespectrum}
If $\Delta_n$ is the graph Laplacian of the $n$th approximating graph to a $3N$-Gasket, then $\Delta_n$ is a matrix of dimension $\displaystyle\frac{3N+(3N-2)(3N)^n}{3N-1}$ and the set of eigenvalues is
\begin{align*}
\set{0,3/2} &\cup \bigcup_{m=0}^{n-1}R^{-m}\paren{
\mathcal A \cup \set{\sin^2\paren{j\pi/(3N)}}_{j=1}^{3N-1} }
\end{align*}
where
\[
\mathcal A = 
\set{z:2 T_N(1-2 z)+2 U_{N-1}(1-2 z)+1=0}.
\]

The multiplicities of these eigenvalues is as follows
\begin{enumerate}[(i)]
\item The multiplicity of $0$ is $1$, and the multiplicity of $3/2$ is $\displaystyle \frac{3N + (6N-3)(3N)^n}{3N-1}$.


\item For any $n\geq 1$ and $0 \leq m < n-1$ such that $\displaystyle R^m(z)= \sin^2(k\pi/(3N)) \in \sigma(\Delta_1)$  $k=1,2,\ldots 3N-1$, then the multiplicity of $z$ is  
\[
\frac{3N+(3N-2)(3N)^{n-m-1}}{3N-1} \quad\text{ if $3\nmid k$ or  }\quad (3N)^{n-m-1}+1 \text{ if }3\mid k.
\]
\item For any $z$ with $R^m(z) \in \mathcal A$, $n\geq0$, and $0\leq m<n-2$, the mulitiplicity of $z$ is $\displaystyle\frac{(3N)^{n-m-1}-1}{3N-1}$.
\end{enumerate}
\end{thm}

The proof of this theorem is given at the end of Subsection~\ref{multiplicities}.

We define $\tilde R$ to be the branch of the inverse of $R(z)$ such that $\tilde R(0) = 0$, that is $\tilde R\circ R(z) = z$ for $z$ in a neighborhood of $0$. Using the above theorem and equation \ref{LaplacianDefn}, one has

\begin{thm}\label{lambda-thm}
If, as before, $\Delta = \lim_{n\to\infty} c^n\Delta_n$ is the Laplcian on $K$, with $c = 3N+2N^2$, then the spectrum
\[
\sigma(\Delta) = \set{\lim_{n\to\infty} c^{m+n}\tilde R^{n}(z)~|~z\in\sigma(\Delta_m)}.
\]
\end{thm}

This follows from calculations similar that in \cite[Chapter 3]{Str06}, and is a technique which goes back to \cite{FS92}.

We also are able to calculate the normalized limiting distribution of eigenvalues, that 
is, 
the limit as $n\to\infty$ of the normalized probability measures $\kappa_n$ 
 defined to be  
\[
\kappa_n(\set{z}) = \frac{\mult n{z} (3N-1)}{3N + (6N-3)(3N)^n},
\] 
where $\mult n{z}$ is the multiplicity of $z$ as an eigenvalues of $\Delta_n$. 
The following results is implied by Theorem~\ref{finitespectrum}.

\begin{cor}\label{cor-kappa}  Normalizing with by the number of eigenvalues including multiplicity, the limiting distribution of eigenvalues (the integrated density of states) is a pure point measure $\kappa=\lim\limits_{n\to\infty}\kappa_n$ with the set of atoms
\begin{eqnarray*}
 \set{\frac{3}{2}}\cup \bigcup_{k=0}^{\infty}R^{-k}\Big(\mathcal A \cup (\set{\sin^2\paren{j\pi/(3N)}}_{j=1}^{3N-1} \Big)
\end{eqnarray*}
The value of $\kappa$ at these atoms is given as follows
\begin{enumerate}[(i)]
\item $ \displaystyle \kappa(\set{3/2})=\frac{3N-2}{6N-3}$.
\item For any $m\geq 0$, if $R^m(z) = \sin^2\paren{k\pi/(3N)}$, then
\[
\kappa(z) = \begin{cases}
\displaystyle\frac{3N-2}{6N-3}(3N)^{-m-1} & \text{ if $3\nmid k$}\\[12pt]
\displaystyle \frac{3N-1}{6N-3}(3N)^{-m-1} & \text{ if $3\mid k$.}
\end{cases}
\]
\item For $z$ with $R^m(z)\in \mathcal A$ with $m\geq 0$, $\kappa(z) = \displaystyle (3N)^{-m-1}(6N-3)^{-1}$.
\end{enumerate}  
\end{cor}

In   the previous results, we assumed that $N$ is fixed, and considered limits as $n\to\infty$. 
There are two different kind of limits. One corresponds to the graph Laplacians 
$\Delta_n$, that converge as $n\to\infty$ to a Laplacian on an infinite graph, with limiting density of  
 states $\kappa$. 
The other limit corresponds to the limits of renormalized Laplacian $c^n\Delta_n$, that converge to 
the continuous Laplacian $\Delta$ on our fractal, $3N$-gasket. Eigenvalues of $\Delta$ are given in 
Theorem~\ref{lambda-thm}. 
Our results allow one to explicitly compute the limits     of these objects as $N\to\infty$. 

\begin{cor} \label{cor-N-to-infty}
\begin{enumerate}
	\item   $\displaystyle \lim_{N\rightarrow\infty}\kappa$ is a probability measure that is
	absolutely continuous on the interval $[0,1]$ and has an atom of size $\frac{1}{2} $ at the point $z=\set{3/2}$. 
	\item
	as $N\to\infty$, each postive eigenvalue $\lambda_k$ of $\Delta$ eventually has multiplicity two and 
	 $$\displaystyle \lim_{N\rightarrow\infty}\lambda_k=\frac29\pi^2k^2.$$
\end{enumerate}
\end{cor}

\begin{proof}
(1)  Consider the   circular sawtooth graph, $V_1(N)$, 
	(see Figures~\ref{fig-graphs} and \ref{Sawtooth Graph}
	), 
	and its spectral decimation with respect to the circular graph of $3N$ vertices. 
	This means that we consider the process of spectral decimation that geometrically 
	can be described as removing the teeth from the sawtooth graph. The sawtooth spectral decimation function $R_{sawtooth}(z)$ is computed in lemma \ref{M1} 
	to be $$R_{sawtooth}(z)=2z. $$	The density of states of the probabilistic Laplacian on the circular graph of $3N$ vertices 
	converges as $N\to\infty$ to an absolutely continuous measure on $[0,2]$. By applying 
	$R_{sawtooth}(z)$ to this continuous measure, and using Corollary~\ref{cor-kappa}, we obtain the result. 

(2) The eigenvalues of the circular graph of $3N$ vertices are $$ 1-\cos\left(\frac{2\pi k}{3N}\right)\approx\dfrac2{9N^2}\pi^2k^2$$
	and so the eigenvalues of the circular sawtooth graph, using the function $R_{sawtooth}(z)$ as in lemma \ref{m1spec}, are 
	$$ \frac12\left(1-\cos\left(\frac{2\pi k}{3N}\right)\right)\approx\dfrac1{9N^2}\pi^2k^2.$$
	They have to be multiplied by $c = 3N+2N^2$, which yield the result, taking into account that the 
	function $R(z)$ is approximately linear in the neighborhood of zero. 
\end{proof}

\begin{rmk} 
Following arguments in Corollary~\ref{cor-N-to-infty}, one can show that, as $N\to\infty$, the eigenfunctions converge uniformly to the usual $\sin$ and $\cos$ eigenfunctions on the circle, as our $3N$-fractals $K$ converge, in the Gromov-Hausdorff sense, to the usual circle.   
See 
\cite{BCF+07,Post}
for some related results. 
\end{rmk}

\begin{figure}
\hfill\includegraphics[width=\textwidth]{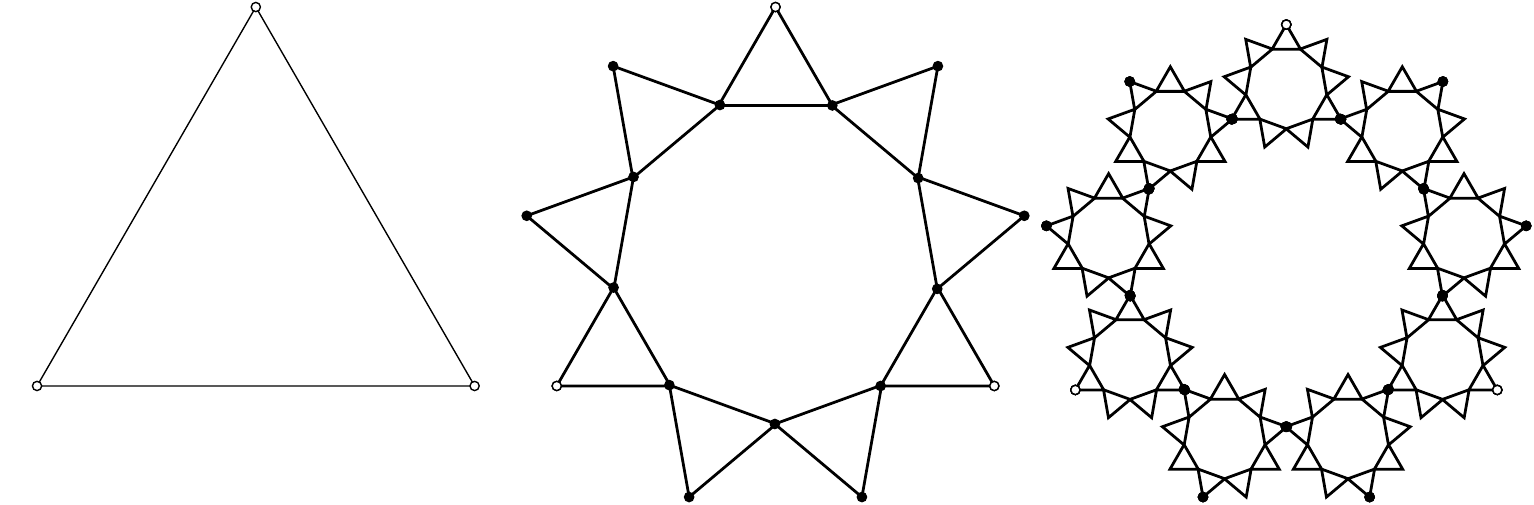}\hfill~
\caption{The 9-gasket graph approximations.\label{fig-graphs}}
\end{figure}
%
\section{Metrics, measure, energy and the heat kernel} \label{metric}

In this section we prove Theorem~\ref{HKEthm}, using techniques from \cite{Kig12,BN98}. At the heart of these techniques, is relating a geodesic metric on the space to the resistance metric.  To define the resistance metric, let
\[
\eng_n(f) = \frac12 \sum_{x\sim y} (f(x)-f(y))^2
\]
for $f\in \ell(V_n)$ be the graph energy of $V_n$. It is shown in \cite{BCF+07,Kig01} that 
\[
\eng(f) := \lim_{n\to\infty} \rho^{-n}\eng_n(f|_{V_n}),
\]
for $$\rho = c/(3N)=1+2N/3$$ and functions $f:K\to\R$, is well defined with non-trivial domain
\[
\mathcal D  : = \set{f:K\to\R~:~\eng(f)< \infty}.
\] 
 It is also established that the Laplacian $\Delta$ defined in Section~\ref{preliminaries} is the infinitesimal generator of $\eng$ as a Dirichlet form on $L^2(K,\mu)$. For more on Dirichlet forms, see \cite{ChenF,FOT11}.

We define the effective resistance between two points $x,y\in K$ as follows
\[
R(x,y) := \inf\set{ \frac{(f(x)-f(y))^2}{\eng(f,f)}~:~f\in\mathcal D}.
\]
It is also established \cite{Kig01} that $R(x,y)$ is a metric on $K$. 

The reader should not confuse this resistance metric $R(x,y)$ with the rational function 
$R(z)$ heavily studied in our  paper but not in the current section. Notation 
for $R(x,y)$ and 
$R(z)$ is well-established in the literature, and we do not to wish to 
change the tradition, even though it may be confusing within one paper. 

We now define the geodesic metric on our fractal spaces.

\begin{prop}\label{metricexistance}
There exists a metric $d_*$ on $K$, inducing the original topology, with the following properties

\begin{enumerate}[(i)]
\item $d_*(x,y) = (N+1)^{-n}d_n(x,y)$ for all $x,y\in V_n$, where $d_n(x,y)$ is the graph distance in $V_n$ defined in \eqref{e-dn}.

\item For all $w\in \alb^n$, $d_*(F_w(x),F_w(y)) = (N+1)^{-n}d_*(x,y)$.

\item \label{geoasmp}
If   $v\neq w\in \alb^n$ and $F_w(K)\cap   F_v(K)=\varnothing $, then 
\[
d_*(x,y) \geq (N+1)^{-n}\quad\text{for all}\quad x\in F_w(K),\ y\in F_v(K).
\]
On the other hand, if $w\in \alb^n$,
\[
d_*(x,y) \leq c(N+1)^{-n}\quad\text{for all}\quad x,y\in F_w(K),
\]
for some constant $c \leq 3 + 1/N$.

\item $K$ is a geodesic space with $d_*$ as a metric.

\item The Hausdorff dimension of $K$ with respect to $d_*$ is
\[
\dim_H(d_*) = \frac{\log 3N}{\log (N+1)}.
\]
\end{enumerate}
\end{prop}

\begin{proof}
If $x\sim y$ in $V_n$, $d_{n+1}(x,y) = d_1(v_0,v_1) = N+1$, because of the self-similarity of the graph. More generally, if $x,y\in V_n$ with $d_n(x,y) = k$, and $x=x_0,x_1,\ldots,x_k = y$ is a length minimizing path, then we can extend this to a path in $V_{n+1}$ by paths which connect $x_i$ to $x_{i+1}$ of length $N+1$.

This path can be seen to be the shortest path because $F_w(V_1)\cap F_v(V_1)\subset V_n$ for $v\neq w \in \alb^n$, thus any path in $V_{n+1}$ between $x$ and $y$ would induce a path in $V_n$ by only considering elements in that path which are also in $V_n$. Thus $d_{n+1}(x,y) = (N+1)d_n(x,y)$ for all $x,y\in V_n$. Thus the metric $d_*$ on $V_* := \cup_{n=0}^\infty V_n$ by $d_*(x,y) = (N+1)^{-n}d_n(x,y)$ for $x,y\in V_n$ is well defined.

We shall prove (i-iii) for $d_*$ defined on $V_*$, which will imply that $d_*$ is jointly continuous with respect to the subspace topology of $V_*\subset K$. This in turn implies that $d_*$ extends to a metric on $K$ which induces the original topology, and that (i-iii) are satisfied by the extended metric. The fact that $K$ is compact proves that $d_*$ is complete.

We have taken (i) as a definition, and (ii) can be seen inductively from the argument above. It is left to prove (iii). 

Since the $F_i$ are $d_*$-similitudes with scaling factor $(N+1)^{-1}$, it is enough to show that the above is true with $n=0$. First, given $y\in K$, for all $k$, there is $v\in \alb^k$ such that $y\in F_v(K)$. Choose $y_k\in F_v(V_0)$. Now,
\[
d_*(y_k,y_{k+1}) \leq \operatorname{diam}V_1/(N+1)^{k+1}
\]
where $\operatorname{diam}V_1 =\lfloor 3N/2 \rfloor + 1$ is the graph diameter of $V_1$. So
\[
d_*(v_i,y) \leq \operatorname{diam}V_1\sum_{k=1}^\infty \frac1{(N+1)^k} = \frac{\operatorname{diam}V_1}{N}.
\]
Thus by the triangle inequality, we have our bound on $d_*(x,y)$. 

(iv) For points on $V_*$, it is easy to see that $d_*$ has approximate midpoints: i.e., for every $x,y\in V_*$, and $\eps>0$, there is $z\in K$ such that $|d(x,y)/2-d(y,z)|\leq \eps$. In this case we can take $z$ from $V_k$ for $k$ large enough. By the standard theory of metric spaces, see \cite{BBI01}, this proves that $d_*$ is a geodesic metric.

(v) follows because $K$ is the attractor of the iterated function system generated by the $3N$ $F_i$, which are each similitudes with respect to $d_*$ with Lipshitz constant $1/(N+1)$.
\end{proof}

\begin{lem}\label{resasmp}
\begin{enumerate}[(i)]
\item There is $c_1$ such that if $x,y\in F_w(K)$ for $w\in \alb^n$, then 
\[
R(x,y) \leq c_1\resscaling^{-n}
\]
\item There is a $c_2$ such that for any $x\in V_n$ and $y\in F_w(K)$ with $w\in \alb^n$ and $x\notin F_w(K)$, then 
 \[
R(x,y) \geq c_2\resscaling^{-n}.
\]
\end{enumerate}
\end{lem}

\begin{proof}
Part (i) follows from proposition 7.16 (b) in \cite{BN98} which states there is a constant $c_1$ such that $|f(x)-f(y)|^2 \leq c_1\rho^{-n}\eng(f)$ for $x,y\in F_w(K)$, and then using s proof similar to that of part (\ref{geoasmp}) of proposition \ref{metricexistance}.

Part (ii): There are at most two $w$ in $\alb^n$, such that $x\in F_w(V_0)$. We construct the function $h$ such that $h\circ F_w \equiv 0$ for $w$ such that $x\notin F_w(K)$, and if $x\in F_w(K)$, then $F_w^{-1}(x)$ is in $V_0$, we define $h\circ F_w$ to be the harmonic extension of the function $g$ which is $1$ at $F^{-1}_w(x)$ and $0$ at the other boundary points. Noting that $\eng(g) = \resscaling/2$, then 
\[
R(x,y) \geq \eng(h) = \resscaling^{-n+1} \quad\text{or}\quad \resscaling^{-n+1}/2
\]
\end{proof}

\begin{prop}\label{snowflake}
There is a constant $c$ such that
\[
c^{-1} (d_*(x,y))^\alpha \leq R(x,y) \leq c (d_*(x,y))^\alpha
\]
for all $x,y\in K$ and $\alpha = \log\resscaling/\log(N+1)$.
\end{prop}

\begin{proof}
Define 
\[
N_n(x) = \cup\set{F_w(K)~|~ F_w(K)\cap F_v(K)\neq \emptyset~\text{for $v$ such that}~x\in F_v(K)}
\]
to be the union of $n$-cells which contain $x$ or intersect a cell which contains $x$. If $y\in N_n(x)$, then by proposition \ref{metricexistance} (\ref{geoasmp}) $d_*(x,y) \leq 2c_1(N+1)^{-n}$. Similarly, $R(x,y)\leq 2c_2\resscaling^{-n}$.

On the other hand, if $y\notin N_n(x)$, any path from $x$ to $y$ must cross two elements of $V_n$ (possibly including $x$). We conclude $d_*(x,y) \geq c_3(N+1)^{-n}$. The function defined in the proof of \ref{resasmp} vanishes outside of $N_n(x)$, so $R(x,y) \geq c_3\resscaling^{-n}$.
\end{proof}

\begin{proof}[Proof of Theorem \ref{HKEthm}]
We shall walk through the requirements of Theorem 15.10 of \cite{Kig12}. The property (ACC) is implied by $R(x,y)$ being uniformly perfect for local Dirichlet forms, see \cite[Proposition 7.6]{Kig12}. Our Dirichlet form is local, and $R(x,y)$ is uniformly perfect which is implied by proposition \ref{snowflake} and the fact that $d$ is geodesic. The fact that $d\sim_{QS} R$ is also implied by proposition \ref{snowflake} with $g(r) = r^{d_w}$, where $d_w = \alpha + d_H$ and $d_H = \operatorname{dim}_H(d_*)$. It also follows from proposition \ref{metricexistance} part (iii) that there are constants $c_5$ and $c_6$ such that
\[
c_5 d_*(x,y)^{d_H}\leq \mu(B_{d_*}(x,d_*(x,y))) \leq c_6 d_*(x,y)^{d_H},
\]
which along with proposition \ref{snowflake} proves the (DM2)$_{g,d_*}$. The chain condition is also implied by the fact that $d_*$ is geodesic. Thus Theorem \ref{HKEthm} follows from \cite[Theorem 15.10]{Kig12} using condition (b).
\end{proof}
%
\section{Computation of $R(z)$, $\phi(z)$ and the exceptional set}\label{computation}

Consider the sawtooth graph $G_2$ with vertex set $\set{v_0,\ldots,v_m,u_1,\ldots,u_{m}}$ and edge relation $v_i \sim v_{i+1}$, $v_i\sim u_i$, and $v_i \sim u_{i+1}$ for $i=0,1,\ldots, m-1$. Let $\Lap_2$ be the graph Laplacian of $G_2$. The sawtooth graph is depicted in Figure \ref{Sawtooth Graph}.

We consider the eigenvalue problem with boundary $v_0$ and $v_m$, that is to say that to find $f\in\ell(G_2)$ and $z\in\C$ such that
\[
\Lap_2 f(x) = zf(x) \text{ for all } x\neq v_0,v_m,
\]
given prescribed values for $f(v_0)$ and $f(v_m)$.

\begin{figure}[t]
\includegraphics[scale=1]{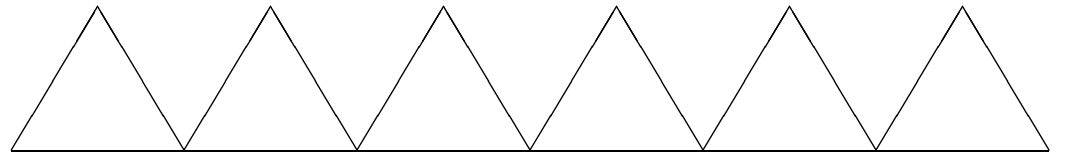}
\caption{Sawtooth graph $G_2$ }\label{Sawtooth Graph}
\end{figure}

\begin{lem}\label{M1}
For $0<z<1$, $z \neq (1-\cos(k\pi/m))/2$, the eigenvalue problem with boundary on $G_2$ is solved by linear combinations of
\begin{eqnarray}
f_{1}(v_k) = -\frac{\sin((k-\frac{m}{2})\arccos(1-2z))}{\sin(\frac{m}{2}\arccos(1-2z))} \label{ef_asym}\\
f_{2}(v_k) = \frac{\cos((k-\frac{m}{2})\arccos(1-2z))}{\cos(\frac{m}{2}\arccos(1-2z))} \label{ef_sym},
\end{eqnarray}
where we extend $f$ to $u_k$ by the eigenvalue equation
\[
f_i(u_k) = \frac{f_i(v_k)+f_i(v_{k+1})}{1-z},\quad i=1,2.
\]
For $z=3/2$, the $z$-eigenspace is span by the functions
\[
g_j(u_k) = (-\delta_{jk})^j \quad \text{and}\quad g_j(v_k) = \begin{cases}
(-1)^k & k<j \\
0 &\text{otherwise,}
\end{cases}
\]
for $j=1,2,\ldots,m$ and $\delta_{jk}$ is the Kronecker delta. Along with $g_0(u_k) = 0$ and $g_0(u_k) = (-1)^k$.
\end{lem}
\begin{proof}
First consider the subgraph $G_1$ consisting of the vertexes $v_0,v_1,\ldots,v_m$, and the Laplacian $\Lap_1$ on this subgraph. It is known that functions of the form $g(v_k) = e^{ik\theta}$, for $\theta$ to be determined, have eigenvalues $1-\cos(\theta)$ except, potentially, at the endpoints $v_0$ and $v_m$.  Suppose $f$ is an eigenfunction of $\Lap_2$ with eigenvalue $z_2$, except at $v_0$ and $v_m$. Then, for $k\in\set{1,\ldots,m-1}$,
\begin{eqnarray}
\nonumber &\Lap_2  f(u_{k}) = z_{2}f(u_{k}) = f(u_{k}) - \frac{1}{2}(f(v_{k-1}) + f(v_{k})) \\
\nonumber &\Lap_2  f(u_{k+1}) = z_{2}f(u_{k+1})= f(u_{k+1}) - \frac{1}{2}(f(v_{k}) + f(v_{k+1}))
\end{eqnarray}

\noindent so that through rearrangement we get
\begin{equation}
(1-z_{2})(f(u_{k}) + f(u_{k+1})) = \frac{1}{2}(f(v_{k-1}) + 2f(v_{k}) + f(v_{k+1})). \label{eqn1}
\end{equation}
Now, for $ v_k$, $k\notin\set{0,m}$,
\begin{equation}
\Lap_2  f(v_k) = z_2 f(v_k) = f(v_k) - \frac{1}{4}\paren{(f(v_{k-1}) + f(v_{k+1})) + (f(u_{k}) + f(u_{k+1}))} \label{eqn2}
\end{equation}
so by dividing by $1-z_2$ in (\ref{eqn1}), substituting the result into (\ref{eqn2}), and collecting our $f(v_k)$'s we get
\begin{equation}
(z_{2}-1+\frac{1}{4(1-z_{2})})f(v_k) = -\frac{1}{4}(1+\frac{1}{2(1-z_{2})})(f(v_{k-1}) + f(v_{k+1})). \label{simple}
\end{equation}
By further simplifying (\ref{simple}), we get
\begin{equation}
 2z_2f(v_k) = \Lap_1 f(v_k). \label{poly}
\end{equation} 
In particular, if $z\neq 3/2$, $f|_{G_1}$ is an eigenfunction of $\Lap_1$ with eigenvalue $z_1 = 2z_2$. But since, for an eigenfunction $f$, the values at $u_k$ are determined by the eigenvalue and the functions value at $v_k$ and $v_{k+1}$, we can assume  $ z_1 = 1-\cos(\theta)$. 

In addition, both $\sin(k\theta)$ and $\cos(k\theta)$ are both eigenfunctions themselves on the line graph with eigenvalues of the same form. It is easy to see the $f_1$ and $f_2$ above are eigenfunctions with eigenvalue $z$ as prescribed, and that this technique does not work for $z$ outside of that range. Further, since $f_1(v_0) = -f_1(v_m) = 1$ and $f_2(v_0) = f_2(v_2) = 1$, any value at the $v_0$ and $v_2$  points can be achieved by a linear combination of these two function.

It is easy to check that the functions $\set{g_j}_{j=0}^m$ are linearly independent and are $3/2$-eigenfunctions of $\Lap_2$ (notice that these function are eigenfunctions on the whole of $G_2$, not just away from boundary).

To prove that these are all such eigenfunctions, possibly with boundary, we proceed by induction on $m$. It is clearly true for $m=1$. Define the subgraph of $G_2$, $G_2^{n} = \set{v_k\text{ or }u_k\in G_2~:~k\leq n}$. If $f$ is an $3/2$ eigenfunction of $G_2$, then $f \pm f(u_m)g_m \pm f(v_m) g_0$, where the $\pm$ is determined by the parity of $m$, is zero away from the vertices of $G_2^{m-1}$, and thus when restricted is a $3/2$-eigenfunction with boundary of $G_2^{m-1}$, and thus must be in the span of $\set{f_j}_{j=0}^{m-1}$ by our induction hypothesis.
\end{proof}


\begin{lem}\label{m1spec}
For a $3N$-gasket, the spectrum of $\Delta_1$, the Laplacian of the first level approximating graph is 
\begin{equation}
\sigma(\Delta_1) = \set{\frac{3}{2}}\bigcup\set{ \sin^2\paren{\frac{j \pi}{3N}} \text{ : } j=0,1,\ldots,3N-1}. \label{Mev}
\end{equation}
where the multiplicity of $3/2$ is $3N$ and where the multiplicity of the other eigenvalues is given by their multiplicity in (\ref{Mev}).
\end{lem}

\begin{proof}
Using terminology from the proof of lemma \ref{M1}, the first level approximating graph of a $3N$-gasket can be thought of as $G_2$ for $m=3N$ and the identification of boundary points $v_0 = v_{3N}$. It is straight forward to check that $f(v_k) = e^{i \theta k}$ is still an eigenfunction of $G_1$ with the aforementioned identification, and that it satisfies the eigenvalue problem at $v_0=v_{3N}$. However, since $f(v_0) = f(v_{3N})$, we have that $\theta = \frac{2j\pi}{3N}$, $j=1,\ldots,3N$.  Eigenfunctions can be extended to all of $G_2$ in the same fashion as without the identification, so by (\ref{poly}) applied with $\theta = \frac{2j\pi}{3N}$, we get eigenvalues of the form (\ref{Mev}).  Note that every eigenfunction on $G_1$, $z_1\neq 3/2$, can be extended to $G_2$ with eigenvalue $(1-\cos\theta)/2 = \sin^2(\theta/2)$.

The other $3N$ eigenfunctions all have eigenvalue $3/2$. The $3/2$-eigenspace is span by, for example $g_j-cg_1$, for $j=0,2,3,\ldots,3N$ where $c\in\set{-1,0,1}$ is chosen so the values at $v_0$ and $v_1$ match. It will be important later to note that all but three of these functions have Dirichlet boundary conditions, that is, they take value $0$ everywhere on $V_0$.
\end{proof}

This linear sawtooth eigenbasis (\ref{ef_asym}) and (\ref{ef_sym}) will be incredibly useful in computing our spectral decimation function $R(z)$.   As it will simplify some calculations later on, we take as shorthand 
\begin{align}
P&=f_1(v_1)+f_1(u_1)=-(f_1(v_{m-1})+f_1(u_{m})) \\ Q&=f_2(v_1)+f_2(u_1)=f_2(v_{m-1})+f_2(u_m). \label{PandQ}
\end{align}
 and we will denote 
\begin{align}
r=\frac{Q-P}{2}\quad\text{and}\quad l=\frac{Q+P}{2}. \label{landr}
\end{align}
That is to say that if $f$ is an eigenfunction of $\Lap_2$ with boundary values $f(v_0)= a$ and $f(v_m) = b$, 
\[
\Lap_2 f(v_0) = a - \frac{al + br}{2} \quad \text{and} \quad \Lap_2 f(v_m) = b-\frac{ar+bl}{2}.
\]

We wish to find expressions for $P$ and $Q$ (and hence $l$ and $r$) in terms of our eigenfunctions.  Accordingly, $\Lap_2 f_1(v_{m}) = z f_1(v_{m}) = f_1(v_{m}) - \frac{1}{2}(1 + f_1(v_{1})) $ in the symmetric case and  $\Lap_2 f_2(v_{m}) = z f_2(v_{m}) = f_2(v_{m}) - \frac{1}{2}(f_2(v_{1}) - 1)$ in the skew-symmetric case.  These expressions give us
\begin{align}
P &= f_1(v_1) + f_1(v_{m}) = f_1(v_1) + \frac{1}{2(1-z)}(1 + f_1(v_1)) \label{P} \\
\nonumber  &= \frac{(3-2z)f(v_1)+1}{2(1-z)} 
\end{align}
and
\begin{align}
Q &= -(f_2(v_{1}) + f_2(v_{m})) = -(f_2(v_1) - \frac{1}{2(1-z)}(- 1 + f(v_1)))  \label{Q} \\
\nonumber  &= -\frac{(3-2z)f_2(v_1)-1}{2(1-z)}.
\end{align}

\subsection{Computation of $R(z)$} \label{sec4a}

Having an eigenbasis and a set of eigenvalues, we now wish to find the spectral decimation function on the $3N$-Gasket, which we will call $R(z)$.  
Let $g$ be an arbitrary eigenfunction of $\Delta_{n+1}$ on the sawtooth graph with eigenvalue $z$, and shall assume a priori that $g|_{V_n}$ has eigenvalue $R(z)$ (which will turn out to be expressed as a rational function of $z$). As we wish to look at the restriction to $V_n$. We consider two adjacent $n$-level cells isomorphic to $V_1$ cells which intersect at a point which we will call $x_{3N}$, assign it the eigenfunction the value $g(x_{3N}) = A$  at this vertex.  Moving counterclockwise, we sequentially call the other two points in $V_0$ $x_{4N}$ and $x_{5N}$ and assign the value $B$ to each. We can assume this symmetry because we will not be checking the the eigenvalue condition at $x_{4N}$ and $x_{5N}$, and thus can replace the values of $g$ there with there averages.  

We label the vertices in $V_1$ of degree $4$, $x_0,...x_{3N-1}$, beginning at the vertex immediately counterclockwise from $x_{3N}$ and continuing counterclockwise. Similarly, label the degree 2 points $x_{3N},x_{3N+1},\ldots,x_{6N-1}$.  By symmetry we can assume that $g$ is symmetric under exchanging the two $n$-cells adjacent to $x_{3N}$.

\begin{figure}[t]
\includegraphics[scale=.75]{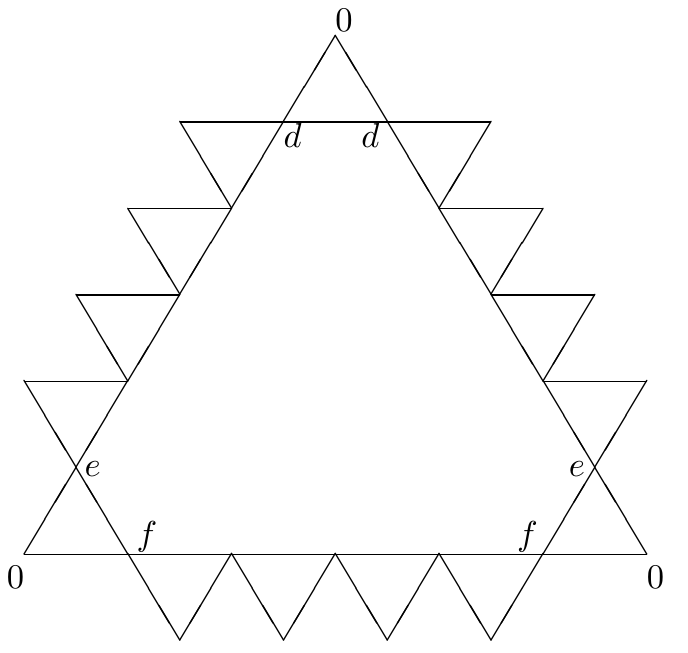}
\caption{ The function $h$.}\label{SawtoothSetup}
\end{figure}

To simplify things, we consider a function $h\in \ell(V_1)$ to locally build $g$ out of, by defining $h(x_{3N}) = 1$ and $h(x_{4N}) = h(x_{5N}) = 0$ and which is extended to the rest of $V_1$ by the relation $\Delta_1 h(x_j) = z h(x_j)$ when $j\neq 3N,4N,5N$. Further we can assume $h$ is symmetric such that $h(x_0)=h(x_{3N-1}) = a$, the vertices adjacent to $x_{4N}$ and $x_{5N}$ have value $h(x_{N-1})=h(x_{2N}) = b$, and the last two vertices adjacent to $x_{4N}$ and $x_{5N}$ $h(x_{N})=h(x_{2N-1})=c$. The values of $g$ are shown in figure \ref{SawtoothSetup}.  

With this scheme we get
\begin{align*}
\Delta_{1} h(x_0) = z a = a - \frac{1}{4}(1 + a + h(x_1) + h(x_{3N+1})) = a - \frac{1}{4}(1 + a + ar+bl) 
 \\
\Delta_{1} h(x_{n-1}) = z b = b - \frac{1}{4}(0 + c + h(x_{N-2}) + h(x_{4N-1}))= b - \frac{1}{4}(c + br+al) 
 \\
\Delta_{n+1} h(x_N) = z c = c - \frac{1}{4}(0 + b + h(x_{N+1}) + h(x_{4N+1})) = c - \frac{1}{4}(b + cr+cl). 
\end{align*} 
Simplifying

\begin{eqnarray}
4(1-z)a = 1 + a + ar + bl \label{eq0} \\
\nonumber 4(1-z)b = c + br + al \\
\nonumber 4(1-z)c = b + cr + cl.
\end{eqnarray}


Returning our focus to the four vertices adjacent to $x$ in $V_n$, we see that $\Delta_n g|_{V_n}(x) = R(z) A = A - \frac{1}{4}(4B) $.
On the other hand,  taking linear combinations of rotations of $h$, we get the value $aA + (b+c)B $ at vertices adjacent to $x_{3N}$ in $V_{n+1}$. Thus $z A = A - (aA + (b+c)B) $.  Combining this with $(1-R(z))A = B $, we get $ z A = A - A(a + (b+c)(1-R(z))) $ which yields 
\begin{equation}
R(z)  = \frac{a + b + c + z - 1}{b+c} = 1+\frac{a+z-1}{b+c} \label{generalR}
\end{equation}

\begin{figure}[b]
  \vspace{-10pt}
  \caption{Derivation of $R(z)$ \label{R_derivation}}
    \vspace{-10pt}
  \centering
    \includegraphics[scale=.75]{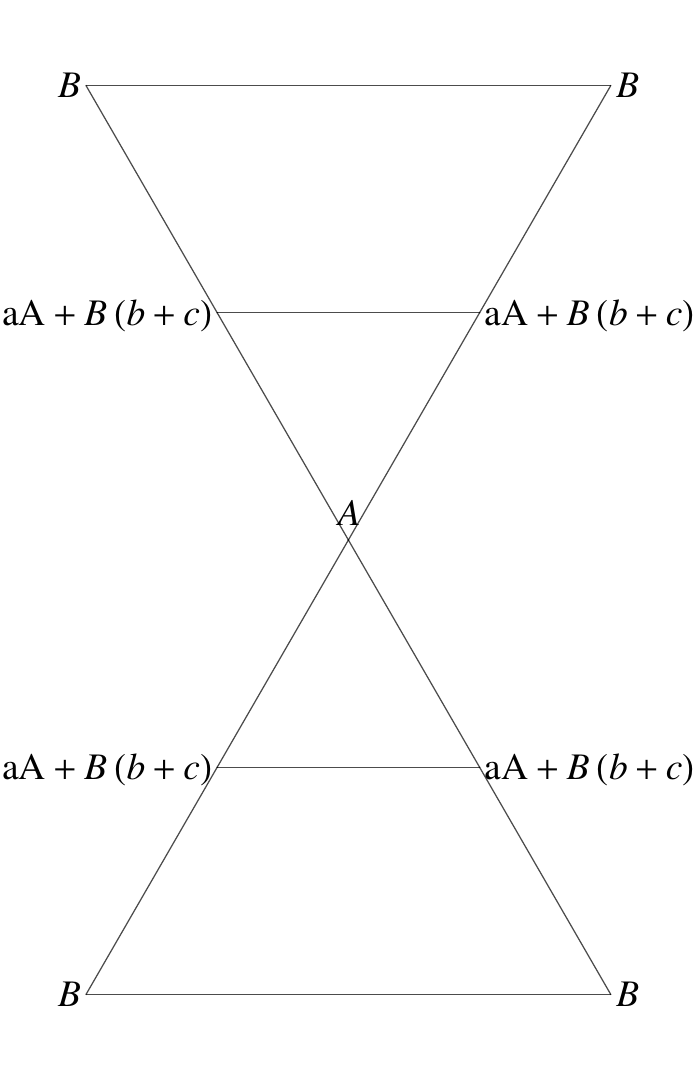}
  \vspace{-25pt}
\end{figure}

We now wish to put $R(z)$ entirely in terms of $z$ (suppressing the $n+1$ subscript).  Thus, we add the three equations in \ref{eq0} and solving for $a+b+c$, getting
\[
a + b + c = \frac{1}{4(1-z) - 1 - r - l} = \frac{1}{3-4z - r - l}.
\]
To calculate $b+c$ in the denominator, we combine equations (\ref{eq0}) to get
\begin{eqnarray}
& (4(1-z)-1-r)(b+c) = (a+c)l \label{eq1} \\
& (4(1-z) - r - l)(a+b) = 1 + (a + c) \label{eq2} \\
& (4(1-z)-r)(a+c) = 1 + (a+b) + (b+c)l \label{eq3}
\end{eqnarray}
Now, solving these three equations for $b+c$ yields
\begin{equation}
 b+c = -\frac{l(-5+l+r+4z)}{(l+r+4 z-3) ((l-1) l-(r+4 z-5) (r+4 z-3))}\label{bplusc}
\end{equation}
and thus
\begin{equation}
R(z) = -\frac{(1+(z-1)(3-l-r-4z)) ((l-1) l-(r+4 z-5) (r+4 z-3))}{l(-5+l+r+4z)}. \label{R}
\end{equation}

With (\ref{R}), (\ref{P}), (\ref{Q}), (\ref{ef_asym}) and (\ref{ef_sym}) we have all of the ingredients we need to construct our function $ R(z)$.  With much computer aid we found the above to be equivalent to the form in Theorem~\ref{Randphi}.

\subsection{Computation of $\phi(z)$}
We would like to compare the $V_1$ approximation to the $V_0$ approximation, so we treat function $g$ of our Laplacian in vector in block form
$$g=\begin{bmatrix}
g_0\\
g_1
\end{bmatrix},$$
where $g_0 = g|_{V_0}$ and $g_1= g|_{V_1\setminus V_0}$. Consider $\Delta_1$ as a $6N\times 6N$ matrix with the above block structure written
$$\Delta_1=\begin{bmatrix}
A&B\\
C&D
\end{bmatrix}$$
where $A$ is a $3\times 3$ matrix, $D$ is a $(6N-3)\times (6N-3)$ matrix and $B$ and $C$ are appropriately sized rectangular matrices.

If we further assume $g$ is an eigenfunction of $\Delta_1$ away from the boundary $V_0$, then
\begin{equation}
\begin{bmatrix}
A&B\\
C&D
\end{bmatrix}
\begin{bmatrix}
g_0\\
g_1
\end{bmatrix}=
\begin{bmatrix}
g_0'\\
zg_1
\end{bmatrix}\label{matrixEq}
\end{equation}

Focusing on the second equation given to us by (\ref{matrixEq}), we see that
\begin{align}
-(D-z)^{-1}Cg_0&=g_1, \label{u1}
\end{align}

Recalling that $S(z)=(A-z)-B(D-z)^{-1}C = \phi(z)(\Delta_n - R(z))$ is the Schur complement of the matrix $\Delta_1-z$, and $S_{1 \ 1}(z)$ is the first entry in the first row of this Schur complement. We compute that 
\begin{eqnarray*}
R(z)&=&1-\frac{S_{1 \ 1}(z)}{\phi (z)}.
\end{eqnarray*}

Now, to find $\phi(z)$ we must first find $S_{1,1}(z)$. Accordingly, we multiply the vector $u_0=(1,0,0)^T$ by $S(z)$, and use (\ref{u1}) and the fact that $A$ is a $3\times 3$ identity matrix to get
$$S_{1 \ 1}(z)=(1-z)+(Bu_0)_1$$

The rows of $B$ correspond to boundary points,
\[
B_{i,j} =
\begin{cases}
-1/2 & \text{if } v_i\sim v_{j+3}\\
0 & \text{otherwise}
\end{cases}
\]

Therefore, the first entry of $Bu_0$ is equal to $-1/2$ times the sum of the value of $u_1$ at each vertex adjacent to the boundary point with value 1. Since these boundary conditions  correspond directly to the eigenbasis we picked for our fractal when calculating $R(z)$, the boundary point in question is only adjacent to two other vertices, both of which have value $a$, using the notation from Subsection~\ref{R_derivation},  $Bu_0=-a.$

Substituting this value into our equality for $S_{1,1}(z)$ we get $S_{1,1}(z)=1-z-a$. Thus, using this fact and (\ref{generalR}) we have
$$R(z)=1-\frac{1-z-a}{\phi(z)} = 1 + \frac{a+z-1}{b+c} .$$
In particular,
$$\phi(z)=b+c,$$
which we had solved for in terms of $l$ and $r$ earlier. Using a computer algebra system, we obtained the simplified version in Theorem~\ref{Randphi}.
\subsection{Poles of $R(z)$ and zeros of $\phi(z)$}

We prove that the singularities of $R(z)$ are of the form \ref{sing}. First, note that $T_N$ and $U_n$ have the $N$ unique roots 
\[
T^{-1}_N(0) = \set{\cos\paren{\frac{(m-\pi/2}{N}}}_{m=1}^N \quad\text{and}\quad U^{-1}_N(0) = \set{ \cos\paren{\frac{m\pi}{N+1}}}_{m=1}^N.
\]
Further, since $\frac{d}{dx}T_N(x) = N U_{N-1}(x)$, the extrema of $T_N(x)$ must be zeros of $U_{N-1}$ and the value at these points must be $\pm 1$, see, for example \cite{MH03}. 

Consider the case when $N$ is even.  Then it is clear from (\ref{form1}) and the previous discussion that the singularities of $R(z)$ must be contained in the ${N}/{2}$ zeros of $T_N(\sqrt{z})$,  
\[
\zeta_m = \cos^2\paren{\frac{(2m-1)\pi}{2N}} = \frac12\paren{1 + \cos\paren{\frac{(2m-1)\pi}{N}}} .
\]
We have
\begin{eqnarray}
& T_N(1-2\zeta_m) = T_N(-\cos(\frac{(2m-1)\pi}{N})) = T_N(\cos(\frac{(2m-1)\pi}{N})) \label{Teven} \\
& U_{N-1}(1-2\zeta_m) = U_{N-1}(-\cos(\frac{(2m-1)\pi}{N})) = -U_{N-1}(\cos(\frac{(2m-1)\pi}{N})) \label{Ueven}
\end{eqnarray}
\noindent where negatives are determined in the last equalities depending on the parity of the degree of the Chebyshev polynomial.  But the form of the extrema of $T_N(x)$ shows that (\ref{Teven}) is always equal to $\pm 1$ and the form of the roots of $U_{N}$ indicates (\ref{Ueven}) is always zero.  Thus, 
\[
2(T_N(1-2 \zeta_m)+ U_{N-1}(1-2 \zeta_m))+1 = 2(\pm 1 + 0) + 1 \neq 0
\]
\noindent and $U_{N-1}(\zeta_m) = U_{N-1}(\cos(\frac{(m-\frac{1}{2})\theta}{N})) \neq 0$ as the roots of $U_{N-1}(x)$ are exactly of the form $\cos(\frac{m\pi}{N})$.  Therefore, the numerator of $R(z)$ in the case where $N$ is even is never zero so no singularity in this case is removable. The proof is similar for $N$ odd.

\section{Gaps in the spectrum}\label{gapproof}
In this section we follow \cite{HSTZ} to establish properties of the gaps in the spectrum of the Laplacian on a $3N$-gasket in  the sense of \cite{Str05}. We say that an infinite non-negative increasing sequence $\set{\alpha_i}_{i=1}^\infty$ has gaps, if $\limsup \alpha_{k+1}/\alpha_k > 1$. 

\begin{prop}\label{poles}
Let $z_{max}=\max\set{z~:~R(z) = z\text{~or~}R(z)=0}$ and $I_0= [0,z_{max}]$. If there  is a pole of $R $ in $I_0$, then there are gaps in the spectrum of $\Delta$. 
\end{prop}

\begin{proof}
The Julia set of $R(z)$, call it $\mathcal{J}$, is contained in in $R^{-n}(I_0)$, for $n=0,1,2,\ldots$. Because there is a pole in $I_0$, $R^{-n}(I_0)\subsetneq R^{-n+1}(I_0)$, and is a finite set of intervals. Since $\mathcal{J}\neq I_0$, it must be completely disconnected.

This implies that for every $\eps\geq 0$, there is a sub interval of $[0,\eps]$ which is not contained in $\mathcal J$. Which implies that there are gaps in the spectrum of $\mathcal{J}$ by an argument similar to the un-numbered Theorem in Section~3 of \cite{HSTZ}.
\end{proof}

\begin{figure}[t]
\includegraphics[width=.49\textwidth]{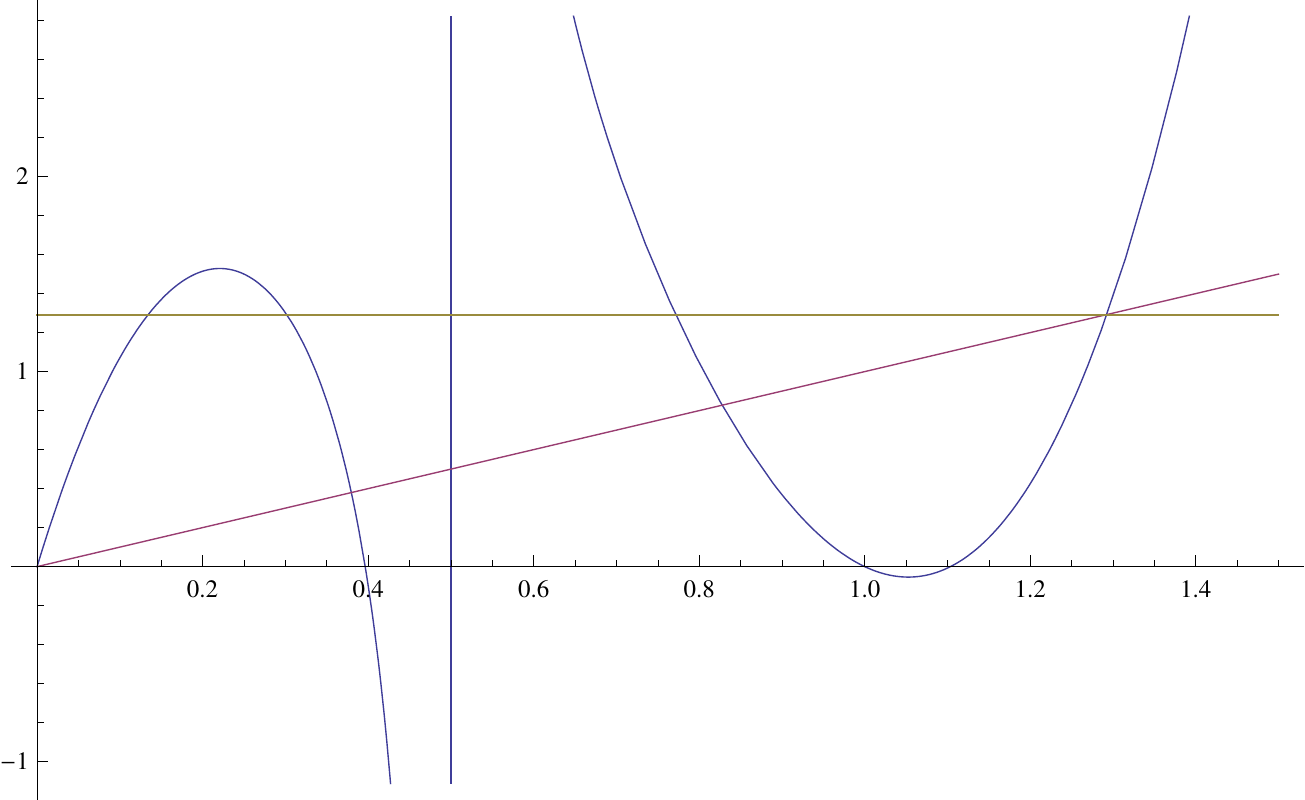}
\includegraphics[width=.49\textwidth]{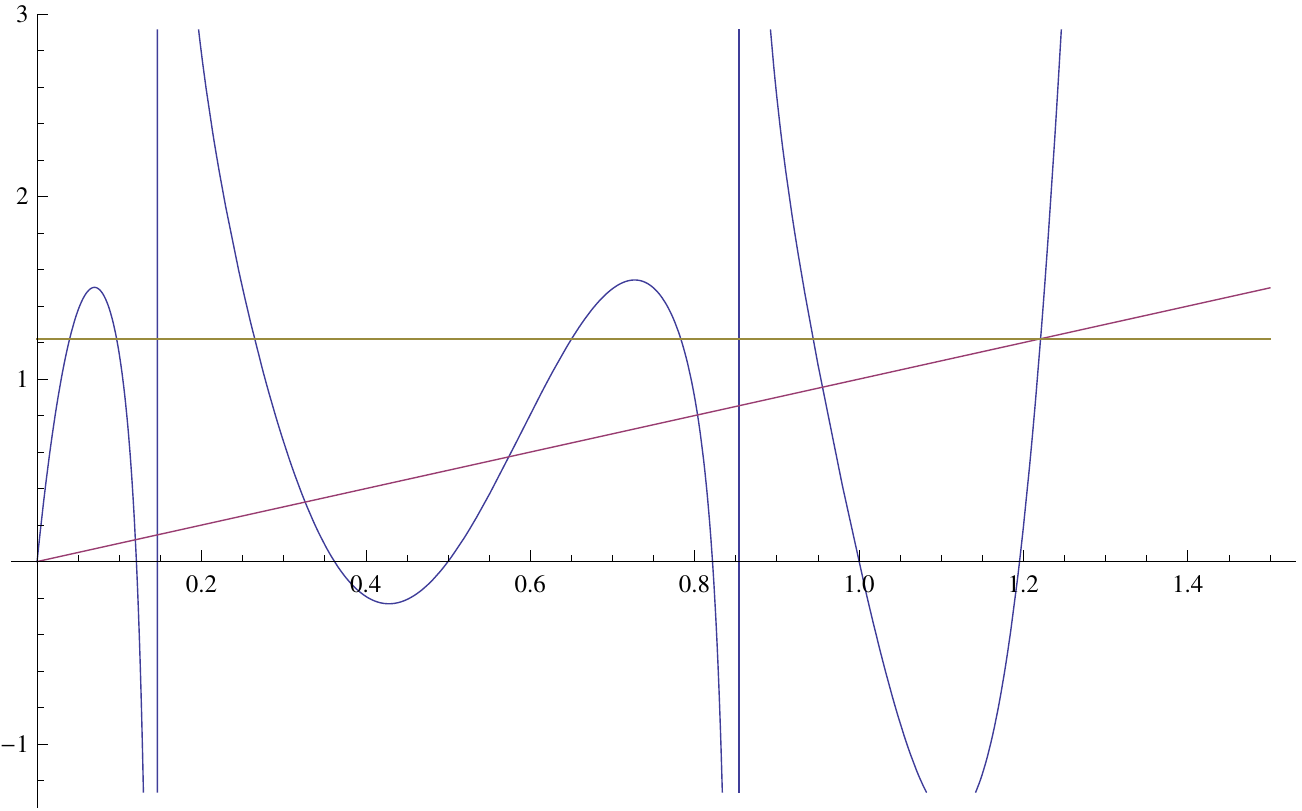}
\caption{$R(x)$ for $N=2$ (left) and $N=4$ (right).}
\end{figure}

\begin{proof}[Proof of Theorem \ref{gapsthm}]
From Theorem \ref{Randphi}, the poles of $R(z)$ are calculated to be real and between $0$ and $1$. 
Also, we see that there are zeros of $R(z)$ at the point, includes $1$ when $N$ is even and $\cos^2(\pi/2N)$ if $N$ is odd. Either way there is a zero of $R(z)$ greater than the largest pole. Thus by Theorem \ref{poles} there are gaps in the spectrum.
\end{proof}
\section{Spectral decimation of $3N$-gaskets} \label{theorems}
\subsection{The exceptional set}

To  compute the multiplicities of eigenvalues of the approximating graph, we must know   the exceptional set  $E(\Delta_0,\Delta) = \sigma(D) \cup \set{z:\phi(z) = 0}$. Here $D$ is the matrix from the block representation of
\[
\Delta_1=\begin{bmatrix}
A&B\\
C&D
\end{bmatrix},
\]
and $\sigma(D)$, the eigenvalues of $D$, correspond to the solutions to the eigenvalue problem on $V_1$ with Dirichlet boundary conditions at $V_0$. That is, if $\lambda \in \sigma(D)$, then there is $f\in \ell(V_1)$ such that $f|_{V_0} = 0$ and $\Delta_1f(x) = \lambda f(x)$ for $x\in V_1\setminus V_0$. Written as a block matrix,
\[
f = \begin{bmatrix}
0\\ g
\end{bmatrix}
\]
where $Dg = \lambda g $, because
\[
\Delta_1 f= \begin{bmatrix}
A& B\\ C&D
\end{bmatrix}
\begin{bmatrix}
0\\ g
\end{bmatrix}
 = 
 \begin{bmatrix}
Bg \\ Dg 
\end{bmatrix}.
\]
\begin{prop}
The exceptional set
\begin{align}
E(\Delta_0,\Delta)=\set{\frac{3}{2}}\cup \mathcal A \cup \mathcal B&\cup\set{\cos^{2}\paren{\frac{m\pi}{2N}}:m=1,...,N-1},
\end{align}
where
\[
\mathcal A =\cup\set{z:2 T_N(1-2 z)+2 U_{N-1}(1-2 z)+1=0},
\]
and
\[
\mathcal B = \begin{cases}
\set{z:T_N(\sqrt{z})-2 (z-1) \sqrt{z} U_{N-1}(\sqrt{z})=0} & \text{ if $N$ is even} \\
\set{z:U_{N-1}(\sqrt{z})- 2 \sqrt{z} T_N(\sqrt{z})=0} & \text{ if $N$ is odd}.
\end{cases}
\]
\end{prop}

In Theorem~\ref{Randphi}, it is shown that 
\[
\phi^{-1}(0) = \begin{cases}
\cos^{2}\paren{\frac{m\pi}{2N}},~m=1,3,\ldots,N-1 & \text{if }N\text{ is even, and}  \\
\cos^{2}\paren{\frac{m\pi}{2N}},~m=2,4,\ldots,N-1 & \text{if }N\text{ is odd}.
\end{cases}
\]
What is left is to show the following lemma.

\begin{lem}
The spectrum of $D$ is the union of the set of zeros of $R(z)$, the set of poles of $\phi(z)$, and $\frac{3}{2}$.  Specifically, if $N$ is even
\begin{align*}
\sigma(D)=\set{\frac{3}{2}}&\cup \set{\cos^{2}\paren{\frac{m\pi}{2N}}:m=2,4,...,N}\cup\mathcal A \cup \mathcal B  
\end{align*}
and if $N$ is odd,
\begin{align*}
\sigma(D)=\set{\frac{3}{2}}&\cup \set{\cos^{2}\paren{\frac{m\pi}{2N}}:m=1,3,...,N}\cup\mathcal A \cup \mathcal B  
\end{align*}
%
Here the dimension of the eigenspace of the eigenvalues $3/2$ is $3N-3$, eigenvalues in $\mathcal A$ is 2, and the multiplicity of the rest is counted by multiplicity of occurrences in the set.
\end{lem}

%
%
%
First we classify the eigenvalues of $D$ which are also Eigenvalues of $\Delta_1$, which we refer to as Dirichlet-Neumann eigenvalues.
\begin{lem}
The set of Dirichlet-Neumann spectrum of $\Delta_1$ is 
\[
\sigma(D) \cap \sigma(\Delta_1) = \set{\frac{3}{2}}\cup\set{\cos^2(m\pi/N):m\in\Z,0<m<N/2}.
\]
Here $3/2$ has multiplicity $3N-3$, and the rest have multiplicity according to their number or representations in $\set{\cos^2(m\pi/N):m\in\Z,0<m<N/2}$.
\end{lem}

\begin{proof}
In the proof of lemma \ref{m1spec}, we have already shown that $3/2$ is in $\sigma(D)$, and has multiplicity at least $3N-3$. By a similar argument as that employed at the end of the proof of lemma \ref{M1}, 
it is not difficult to prove that the multiplicity is exactly $3N-3$. 

As for $\cos^2(m\pi/N), m=1,2,\ldots, N/2$, when $N$ is even, and $\cos^2((m-1/2)\pi/N)$. These are roots of $T_N$ and $U_N$ respectively, and are thus zeros of $R(z)$, which means that they are in $R^{-1}(\sigma(\Delta_0))$ as well as poles of $\phi(z)$. 

An element in $\sigma(\Delta_1)$ is either $3/2$ of of the form
\[
\frac{1-\cos\paren{\frac{2j\pi}{3N}}}{2} = \sin^2\paren{\frac{j\pi}{3N} }= \cos^2\paren{\frac{j\pi}{3N}-\frac{\pi}{2}} = \cos^2\paren{\frac{(2j-3N)\pi}{6N}}.
\]
Thus, either
\[
m = \frac{2j-3N}{6} \quad\text{ if }N\text{ is even, or}\quad m=\frac{2j-3(N+1)}{6},\quad\text{ if }N\text{ is odd,}
\]
By choosing appropriate $j$ divisible by $3$ and between $1$ and $3N$ we can find any $m$ between $1$ and $N/2$. Thus all poles of $R(z)$ are in $\sigma(\Delta_1)$. Further, if $z_0\in\sigma(\Delta_1)$ and $z_0\neq 3/2$ then we have shown that the corresponding eigenfunction $f_0$ is an extension of the eigenfunction from the circle. 
\end{proof}


\begin{lem}
The set $\sigma(D)\setminus\sigma(\Delta_1)$ consists of the $N$ elements of $\mathcal A$, each being an eigenvalue with multiplicity 2, and the $\lfloor N/2 \rfloor +1$ elements of $\mathcal B$, each with multiplicity $1$.
\end{lem}

\begin{proof}
 Considering $\Delta_1$ as a block matrix as before
\[
(\Delta_1-z)^{-1}\left[\begin{array}{c}g\\0\end{array}\right] = \left[\begin{array}{c}S^{-1}(z)g \\ -(D-z)^{-1}CS^{-1}(z)g \end{array}\right].
\]
Thus, if $z_0 \notin \sigma(\Delta_1)$ then $\lim_{z\to z_0}S^{-1}(z)$ is well defined. We prove that for $z_0 \notin \sigma(\Delta_1)$ the multiplicity of $z_0$ as an eigenvalue of $D$ is equal to the the dimension of the nullspace $\lim_{z\to z_0}S^{-1}(z)$. 

Since $S(z_0)$ exists and is invertable for $z_0\notin \sigma(D)$, the statement is vacuously true there. On the other hand, if $z_0\in\sigma(D)$ then there is $f\in\ell(V_0\setminus V_1)$ such that
\[
\Delta_1\left[\begin{array}{c}0 \\ f\end{array}\right] = \left[\begin{array}{c}Bf \\ z_0 f\end{array}\right],\quad\text{thus}\quad (\Delta_1-z_0)\left[\begin{array}{c}0 \\ f\end{array}\right] = \left[\begin{array}{c}Bf \\ 0 \end{array}\right].
\]
And in particular $\lim_{z\to z_0}S^{-1}(z)(Bf) = 0$. Thus the matrix $B$ is a linear mapping from the $z_0$-eigenspace of $D$ to the nullspace of $\lim_{z\to z_0} S^{-1}(z))$.

Conversely, assuming $\lim_{z\to z_0}S^{-1}(z)g = 0$ for some $g\in\ell(V_0)$ with $g\neq 0$, because $z_0\notin\sigma(\Delta_1)$ $(\Delta_1-z_0)^{-1}$ is well defined,
\[
\lim_{z\to z_0}  \left[\begin{array}{c}S^{-1}(z)g \\ -(D-z)^{-1}CS^{-1}(z)g \end{array}\right] = \lim_{z\to z_0}  (\Delta_1-z)^{-1}\left[\begin{array}{c}g\\0\end{array}\right] = (\Delta_1-z_0)^{-1}\left[\begin{array}{c}g\\0\end{array}\right] = \left[\begin{array}{c}0 \\ f\end{array}\right].
\]
Here $f\in\ell(V_1\setminus V_0)$ with $f\neq 0$ because $\Delta_1-z_0$ is invertible at $z_0\notin\sigma(\Delta_1)$. Thus, by the block representation we have non-zero $f$ such that $(D-z)f = 0$. Thus the linear map $g \mapsto \lim_{z\to z_0} -(D-z)^{-1}CS^{-1}(z)g$ is a well defined map from the nullspace of $\lim_{z\to z_0}S^{-1}(z)$ to the $z_0$-eigenspace of $D$, and is clearly the inverse of the map above, which proves that the two spaces have the same dimension.


Additionally, for $z\notin\sigma(D)\cup\sigma(\Delta_1)$, $S^{-1}(z)$ can be written
\[
 \frac1{\phi(z)}\paren{\Delta_0-R(z)}^{-1} = \frac{-1}{\phi(z)R(z)(3-2R(z))}\left[\begin{array}{ccc}1-2R(z) & 1 & 1 \\1 & 1-2R(z) & 1 \\1 & 1 & 1-2R(z)\end{array}\right].
\]
Thus
\[
\operatorname{Det}(S^{-1}(z)) = -4(\phi(z))^{-3}(3-2R(z))^{-2}(R(z))^{-1},
\]
So the above is invertible unless  $z_0\in\sigma(\Delta_1)$ or $\phi(z)$ has a pole at $z_0$.

Notice that for $z_0\notin \sigma(\Delta_1)$ and $R(z)\in\sigma(\Delta_0) = \set{0,3/2}$, then because $S^{-1}$ has a removable singularity, $\phi(z)$ must have a pole at $z_0$. Conversely, if $\phi(z)$ has a pole at $z_0$, then either $z\in\sigma(\Delta_0)$ or $S^{-1}(z_0)=0$.  

If $v = [1,1,1]^T$ and $u\in \spn{v}^\perp$, i.e. $v$ is invariant under all symmetries of the $3N$-gasket. It is simple to compute
\[
S^{-1}(z_0) v = \frac{-v}{\phi(z)R(z)} \quad\text{and}\quad S^{-1}(z_0)u = \frac{u}{\phi(z)(3-2R(z))}.
\]

Now assume  $ z_0 \in \mathcal A$ then $R(z_0)= 0$ and $\phi(z)$ has a pole at $z_0$ and $\lim_{z\to z_0}R(z)\phi(z)\neq 0$ and is well defined. Thus $S^{-1}(z_0)u = 0$ and $S^{-1}(z)v\neq 0$, and so the nullspace of $S^{-1}(z_0) = 2$. Thus, $z_0$ corresponds to an eigenspace of dimension 2, span by anti-symmetric eigenfunctions with respect to reflection about a given axis. A geometric argument shows that there should be $N$ such eigenfunctions for a given axis of symmetry (so $2N$ linearly independent functions in total).

On the other hand, for $z_0\in \mathcal B$, then both $\phi(z)$ and $\phi(z)R(z)$ have a pole at $z_0$. Thus $S^{-1}(z_0) v =0$, so $z_0$ corresponds to an eigenfunction which is invariant under all symmetries. A geometric argument shows that there should be $(N+2)/2$ such eigenfunctions if $N$ is even and $(N+1)/2$ if $N$ is odd.

Adding up the multiplicities we have found all of the $3N-3$ eigenvalues of $D$.
\end{proof}

\subsection{Multiplicities}\label{multiplicities}

In this section we recall from \cite{Tep08b} and prove new results about spectral decimation which will allow us to determine the eigenvalues with multiplicities of $\Delta_n$. Using these formulas we will be able to determine the spectrum of the Laplacian on the limiting fractal.

We define $\mult M z$ to be the multiplicity of an eigenvalue $z$ of the matrix $M$, and $ \operatorname{dim}_M$ to be the dimension of a square matrix $M$. Additionally $\mult{n}{z} = \mult{\Delta_n}{z}$ and $\operatorname{dim}_n=\operatorname{dim}_{\Delta_n}$. Using previous work in \cite{Tep08b}, we have the following proposition, letting $m$ denote the number of contractions defining an iterated function system (i.e. for a $3N$-gasket, $m=3N$):

	\def\pro#1{Proposition~\ref{p-mult}(\ref{p-mult#1})}
	\def\Pro{Proposition~\ref{p-mult}}

\begin{prop}\label{p-mult} \cite{Tep08b}
\begin{enumerate}
        \item\label{p-mult1}
If $z\notin E(M_0, M)$, then
\begin{equation}
	\mult n{z}=\mult{n-1}{ R(z)},
\end{equation}
and every corresponding eigenfunction at depth $n$ is an extension of an eigenfunction at depth $n-1$.

        \item\label{p-mult2}
If $z \notin \sigma(D)$, $\phi(z) = 0$ and $R(z)$ has a removable singularity at $z$, then
\begin{equation}
	\mult n{z}=\textup{dim}_{n-1},
\end{equation}
and every corresponding eigenfunction at depth $n$ is localized.

        \item\label{p-mult3}
If $z \in \sigma(D)$, both $\phi(z)$ and $\phi(z)R(z)$ have poles at $z$, $R(z)$ has a removable singularity at $z$,  and $\frac d{dz}R(z)\neq 0$, then
\begin{equation}
	\mult n{z}=m^{n-1}\mult D{z}-\textup{dim}_{n-1}+\, \mult{n-1}{ R(z)},
\end{equation}
and every corresponding eigenfunction at depth $n$ vanishes on $V_{n-1}$.

        \item\label{p-mult4}
If $z \in \sigma(D)$, but $\phi(z)$ and $\phi(z)R(z)$ do not have poles at $z$, and $\phi(z)\neq 0$, then
\begin{equation}
	\mult n{z}=m^{n-1}\mult D{z}+\, \mult{n-1}{R(z)}.
\end{equation}
In this case $m^{n-1}\mult D{z}$ linearly independent eigenfunctions are localized, and $\mult{n-1}{ R(z)}$ more linearly independent eigenfunctions are extensions of corresponding eigenfunction at depth $n-1$.

        \item\label{p-mult5}
If $z \in \sigma(D)$, but $\phi(z)$ and $\phi(z)R(z)$ do not have poles at $z$, and $\phi(z)= 0$, then
\begin{equation}
	\mult n{z}=m^{n-1}\mult D{z}+\, \mult{n-1}{ R(z)}+\, \textup{dim}_{n-1}
\end{equation}
provided $R(z)$ has a removable singularity at $z$. In this case there are $m^{n-1}\mult D{z}+\textup{dim}_{n-1}$ localized and
$\mult{n-1}{ R(z)}$ non-localized corresponding eigenfunctions at depth $n$.

        \item\label{p-mult6}
If $z \in \sigma(D)$, both $\phi (z)$ and $\phi(z)R(z)$ have poles at $z$, $R(z)$ has a removable singularity at $z$, and $\frac d{dz}R(z)= 0$, then
\begin{equation}
	\mult n{z}=\mult{n-1}{ R(z)},
\end{equation}
provided there are no corresponding eigenfunctions at depth $n$ that vanish on $V_{n-1}$. In general we have
\begin{equation}
	\mult n{z}=m^{n-1}\mult D{z}-\textup{dim}_{n-1}+\, 2\mult{n-1}{ R(z)}.
\end{equation}

        \item\label{p-mult7}
If $z \notin \sigma(D)$, $\phi(z)=0$ and $R(z)$ has a pole $z$, then $\mult n{z}=0$ and $z$ is not an eigenvalue.

        \item\label{p-mult8}
If $z \in \sigma(D)$, but $\phi(z)$ and $\phi(z)R(z)$ do not have poles at $z$, $\phi(z)= 0$, and $R(z)$ has a pole $z$, then
\begin{equation}
	\mult n{z}=m^{n-1}\mult D{z}
\end{equation}
and every corresponding eigenfunction at depth $n$ vanishes on $V_{n-1}$.
\end{enumerate}
\end{prop}

In computing the spectrum of the Laplacian of the $3N$ gasket, we shall need the following additional case, which is not covered by the result above.

\begin{thm}\label{thm-new}
Say $z_0\in\mathcal A$, then $R(z_0) = 0$, $\phi(z)$ has a pole at $z_0$ but $\phi(z)R(z)\neq 0$ has a removable singularity, then
\[
\mult n{z_0} = \frac{1-(3N)^{n-1}}{1-3N}
\]
for $n>2$, and $\mult n{z_0} = 0$ otherwise.
\end{thm}

\begin{proof}
In Theorem 3.3 of \cite{Tep08b} it is shown that 
\begin{align*}
(\Delta_n-z)^{-1}& =  (D-z)^{-1}P_n' \\ &+ (P_{n-1}-(D-z)^{-1}C)(\phi(z)\Delta_{n-1}-\phi(z)R(z))^{-1}(P_{n-1}-B(D-z)^{-1}P_n')
\end{align*}
Where $P_{n-1}$ the orthogonal projection on $\ell(V_n)$  to the set of functions supported on $V_{n-1}$, and $P'_n = I_n - P_{n-1}$ is the projection on the orthogonal complement i.e. functions supported away from $V_{n-1}$. Using the spectral decomposition
\[
\Delta_n = \sum_{\lambda\in\sigma(\Delta_n)} \lambda P_{n,\lambda}
\]
where $P_{n,\lambda}$ is the orthogonal projection onto the $\lambda$-eigenspace of $\Delta_n$,  we get 
\begin{align*}
&\sum_{\lambda\in\sigma(\Delta_n)}\frac1{ \lambda-z }  P_{n,\lambda} =  (D-z)^{-1} P_n' \\ 
&+ \sum_{\lambda' \in \sigma ( \Delta_{n-1} ) }(P_{n-1}-(D-z)^{-1}C)\frac{1}{\phi(z)(\lambda'-R(z))}P_{n-1,\lambda'}(P_{n-1}-B(D-z)^{-1}P_n')
\end{align*}
Multiplying both sides by $(z_0-z)$ and taking the limit $z\to z_0$, the left hand side of the above becomes $P_{n,z_0}$ (the projection will be zero when $z_0$ is not an eigenvalue).

The first term on the right hand side becomes
\[
\lim_{z\to z_0} \sum_{\lambda\in\sigma(D)} \frac{z_0-z}{\lambda-z}P_{D,\lambda}P'_n = P_{D,z_0},
\]
where $P_{D,\lambda}$ is the orthogonal projection onto the $\lambda$-eigenspace of $D$. Again not that the projection $P_{D,z_0}$ may be zero. The second term vanishes because
\[
\lim_{z\to z_0}\frac{z_0-z}{\phi(z)\lambda'-\phi(z)R(z)}  = 0,
\]
as $\phi_0(z) :=\phi(z)/(z_0-z)$ has a pole at $z_0$ of order 2. Putting this together,
$P_{n,z_0}  = P_{D,z_0}$.
In other words, if $f$ is an eigenfunction of $z_0$, $f|_{V_{n-1}} \equiv 0$. We establish in the end of Section~\ref{computation} that eigenvectors of $z_0$ for $\Delta_1$ are such that if $u$ is a $z_0$-eigenfunction of $D$ then the values on $V_0$ of $\Delta_1 u$ are orthogonal to constant vectors, i.e.
$
\sum_{x\in V_0} \Delta_1u(x) = 0
$. 
Further, for any vector in $\R^3$ orthogonal to constants, there is a $z_0$-eigenfunction of $D$ such that $\Delta_1 u$ has that vector as boundary values.

In other words, 
 we consider the graph $G$ with vertex set $\set{F_w(V_1)~|~w\in\alb^{n-1}}$, which are the $n-1$ cells of the $V_n$, and edge relation is given by the  non-empty intersections. A $z_0$-eigenfunction $f$ induces a function $J$ on the edges of $G$, by
$
J(F_w(V_1),F_v(V_1)) = -\frac14 \sum_{y\sim x, ~y\in F_w(V_1)} f(y)
$
where $x$ is the unique element in the intersection of $F_w(V_1)$ and $F_v(V_1)$. Further $J$ is such that $J(\alpha,\beta) = -J(\beta,\alpha)$ and $\sum_{\beta\sim\alpha} J(\alpha,\beta) = 0$ (in this sum, $\alpha$ is fixed). Additionally, any function $J$ with these two properties can be obtained in this way. The dimension of the linear space spanned by such functions is the 1st Betti number of the graph, which is easily computed to be
$
1+3N + (3N)^2 + \ldots + (3N)^{n-2} = \dfrac{1-(3N)^{n-1}}{1-3N}
$
\end{proof}

%
\begin{proof}[Proof of Theorem \ref{finitespectrum}]
The dimension of $\Delta_n$ is the same as the number of vertices in $V_n$ which can be calculated by noticing that there are $3N$ critical points --- points where $F_i(K)$ overlaps with $F_j(K)$ where $j\neq i$. Thus the number of vertices satisfies the recurrence relation $|V_{n+1}| = 3N|V_{n}| - 3N$ with $|V_0|= 3$. Note that by results from \ref{p-mult}(\ref{p-mult7}), that every zero of $\phi(z)$ is $3/2$ or not an eigenvalue (has multiplicity $0$).

We calculate the multiplicities (and thus the spectrum) of the eigenvalues below.
\begin{enumerate}[(i)]
\item is found using Proposition \ref{p-mult}(\ref{p-mult5}).  Clearly $\phi\paren{{3}/{2}}=0$ and ${3}/{2}$ is not a singularity of $R(z)$, $\phi\paren{{3}/{2}}R\paren{{3}/{2}}=0$.  Thus, \ref{p-mult}(\ref{p-mult5}) applies and gives us (i), noticing that $R(3/2)$ is not an eigenvalue of $\Delta_{n-1}$ and thus.
\item Let $R^m(z) = \sin^2\paren{k\pi/(3N)}$. If $3\nmid k$, $R^{m}(z)\notin E(\Delta_0,\Delta)$ and thus $R^{m+1}(z)\in\sigma(\Delta_0)$. We can see that know $R^{m+1}(z)\neq 0$, so $R^{m-1}(z)=\frac{3}{2}$. If $3\mid k$,  $R^m(z)\in \sigma(D)$, $R^{m+1}(z)=0$, and because $R^m(z)\in\sigma(\Delta_1)$ and it is not $3/2$,  $\phi(z)\neq 0$, we apply Proposition  \ref{p-mult}(\ref{p-mult4}) . 
Note that If $R(z) = 3/2$ then either $z \sigma(D)$ or $z\in \sigma(\Delta_1)$ then, from Proposition \ref{p-mult}(\ref{p-mult3}), we know that if $z\in \sigma(D)$, then $z$ is not an eigenvalue. Also note, that $z\in \sigma(D)$ implies $z\in \mathcal B$, because all other element of $\sigma(D)$ are $3/2$ or zeros of $R(z)$.
\item If $R^m(z) \in A$, then we apply Theorem \ref{thm-new}.
\end{enumerate}
\end{proof}
%
\bibliography{newbib}
\bibliographystyle{plain}
%
\end{document}